\documentclass{jpp}

\usepackage{color}

\usepackage{graphicx}
\usepackage{natbib}

\usepackage{caption}
\usepackage{subcaption}

\ifCUPmtlplainloaded \else
  \checkfont{eurm10}
  \iffontfound
    \IfFileExists{upmath.sty}
      {\typeout{^^JFound AMS Euler Roman fonts on the system,
                   using the 'upmath' package.^^J}%
       \usepackage{upmath}}
      {\typeout{^^JFound AMS Euler Roman fonts on the system, but you
                   dont seem to have the}%
       \typeout{'upmath' package installed. JPP.cls can take advantage
                 of these fonts, if you use 'upmath' package.^^J}%
       \providecommand\upi{\pi}%
      }
  \else
    \providecommand\upi{\pi}%
  \fi
\fi

\ifCUPmtlplainloaded \else
  \checkfont{msam10}
  \iffontfound
    \IfFileExists{amssymb.sty}
      {\typeout{^^JFound AMS Symbol fonts on the system, using the
                'amssymb' package.^^J}%
       \usepackage{amssymb}%
       \let\le=\leqslant  \let\leq=\leqslant
         
      }{}
  \fi
\fi

\ifCUPmtlplainloaded \else
  \IfFileExists{amsbsy.sty}
    {\typeout{^^JFound the 'amsbsy' package on the system, using it.^^J}%
     \usepackage{amsbsy}}
    {\providecommand\boldsymbol[1]{\mbox{\boldmath $##1$}}}
\fi

\newsavebox{\astrutbox}
\sbox{\astrutbox}{\rule[-5pt]{0pt}{20pt}}

\usepackage{amsmath}
\usepackage{cite}

\newtheorem{theorem}{Theorem}

\title[  On the inverse problem of 1-D collisionless equilibria] {From one-dimensional fields to Vlasov equilibria: Theory and application of Hermite polynomials }

\author[O. Allanson, T. Neukirch, S. Troscheit and F. Wilson]%
{O.\ns A\ls L\ls L\ls A\ls N\ls S\ls O\ls N    %
  \thanks{oliver.allanson@st-andrews.ac.uk},\ns
T.\ns N\ls E\ls U\ls K\ls I\ls R\ls C\ls H  \thanks{ tn3@st-andrews.ac.uk},\ns S.\ns T\ls R\ls O\ls S\ls C\ls H\ls E\ls I\ls T  \thanks{s.troscheit@st-andrews.ac.uk} \break
\and F.\ns W\ls I\ls L\ls S\ls O\ls N \thanks{fw237@st-andrews.ac.uk}   }

\affiliation{School of Mathematics \& Statistics, University of St Andrews,
Fife, KY16 9SS, UK}

\pubyear{2010}
\volume{650}
\pagerange{119--126}
\date{?; revised ?; accepted ?. }
\begin{document}

\maketitle

%\author{O. Allanson, T. Neukirch, S. Troscheit and F. Wilson}
%\address{School of Mathematics and Statistics, University of St Andrews, North Haugh, St Andrews, KY16 9SS, UK}
%\email{oliver.allanson@st-andrews.ac.uk, tn3@st-andrews.ac.uk, s.troscheit@st-andrews.ac.uk, fw237@st-andrews.ac.uk}

\begin{abstract}
We consider the theory and application of a solution method for the inverse problem in collisionless equilibria, namely that of calculating a Vlasov-Maxwell equilibrium for a given macroscopic (fluid) equilibrium. Using Jeans' Theorem, the equilibrium distribution functions are expressed as functions of the constants of motion, in the form of a Maxwellian multiplied by an unknown function of the canonical momenta. In this case it is possible to reduce the inverse problem to inverting Weierstrass transforms, which we achieve by using expansions over Hermite polynomials. A sufficient condition on the pressure tensor is found which guarantees the convergence and the boundedness of the candidate solution, when satisfied. This condition is obtained by elementary means, and it is clear how to put it into practice. We also argue that for a given pressure tensor for which our method applies, there always exists a positive distribution function solution for a sufficiently magnetised plasma. Illustrative examples of the use of this method with both force-free and non-force-free macroscopic equilibria are presented, including the full verification of a recently derived distribution function for the force-free Harris Sheet (\citet{Allanson-2015PoP}). In the effort to model equilibria with lower values of the plasma beta, solutions for the same macroscopic equilibrium in a new gauge are calculated, with numerical results presented for $\beta_{pl}=0.05$. % We also consider the problem of including a spatially inhomogeneous temperature profile. 
\end{abstract}


\section{Introduction}
An important question in the study of plasmas is to understand the fundamental physics involved in magnetic reconnection. Magnetic reconnection processes can critically depend on a variety of length and time scales, for example on lengths of the order of the Larmor orbits and below that of the mean free path (\citet{Biskamp-2000, Birn-2007}). In such situations a collisionless kinetic theory could be necessary to capture all of the relevant physics, and as such an understanding of the differences between using MHD, two-fluid, hybrid, Vlasov and other approaches is of paramount importance, for example see \citet{Birn-2001, Birn-2005} for discussions of this problem in the context of one-dimensional (1-D) current sheets: the `GEM' and `Newton' challenges.

Current sheet equilibria are frequently considered to be the initial state of wave processes, instabilities, reconnection and various dynamical phenomena in laboratory, space and astrophysical plasmas, in theory and observation; see for example \citet{Fruit-2002, Schindlerbook, Yamada-2010}. In particular, force-free current sheets are relevant for the solar corona (\citet{Priest-2000}), Jupiter's magnetotail (\citet{Artemyev-2014}), the Earth's magnetotail (\citet{Vasko-2014,Petrukovich-2015}) and the Earth's magnetopause (\citet{Panov-2011}). Further relevant theoretical works on distribution functions (DFs) for (nonlinear) force-free current sheets are, for example, \citet{Harrison-2009a,Harrison-2009b,Neukirch-2009,Wilson-2011,Abraham-Shrauner-2013,Allanson-2015PoP,Kolotkov-2015}.

In the absence of an exact collisionless kinetic equilibrium solution, one has to use non-equilibrium DFs to start kinetic simulations, without knowing how far from the true equilibrium DF they are. In such cases, non-equilibrium `flow-shifted' Maxwellian distributions are frequently used (see \citet{Hesse-2005,Guo-2014} for examples). Using the DF found in \citet{Harrison-2009b}, the first fully kinetic simulations of collisionless reconnection with an initial condition that is an exact Vlasov solution for a nonlinear force-free field was conducted by \citet{Wilson-2016}.

Motivated by these and other considerations, this paper presents results on the theory and application of a method that allows the calculation of collisionless kinetic plasma equilibria. The method is specifically designed to solve the problem of finding quasineutral collisionless equilibrium DFs, $f_{s}$, for a given macroscopic plasma equilibrium. 

As intimated above, 1-D Cartesian coordinates are very frequently used in the study of waves, instabilities and reconnection (see \citet{Schindlerbook} for example). In this work, $z$ is taken to be the spatial coordinate on which the system depends. Thus the Hamiltonian, $H_s$, and two of the canonical momenta $p_{xs}$ and $p_{ys}$
\begin{eqnarray}
H_s&=&m_{s}\boldsymbol{v}^2/2+q_{s}\phi,\nonumber\\
p_{xs}&=&m_{s}v_{x}+q_{s}A_x,\nonumber\\
p_{ys}&=&m_{s}v_{y}+q_{s}A_y,\nonumber
\end{eqnarray}
are conserved. The particle species is denoted by $s$, with $q_{s}$ the charge, $\boldsymbol{v}$ the velocity and $\phi$ the scalar potential. The vector potential is taken to be $\boldsymbol{A}=(A_x(z),A_y(z),0)$, such that $\boldsymbol{B}=\nabla\times\boldsymbol{A}$. The macroscopic force balance is then given by
\begin{equation}
\frac{d}{dz}\mathsfbi{P}_{zz}=(\boldsymbol{j}\times\boldsymbol{B})_z, \label{eq:motion}
\end{equation}
see e.g. \citet{Mynick-1979a,Harrison-2009a}, with $\boldsymbol{j}=(\nabla\times\boldsymbol{B})/\mu_0$ the current density, $\mu_0$ the magnetic permeability \emph{in vacuo} and $\mathsfbi{P}_{ij}$ the $ij$ component of the pressure tensor
\begin{equation}
\mathsfbi{P}_{ij}=\sum_{s}P_{ij,s}=\sum_{s}m_{s}\int\,w_{is}\,w_{js}\,f_{s}\,d\boldsymbol{v}.
\end{equation}
The particle velocity relative to the bulk is given by $w_{i}=v_{i}-\langle v _i\rangle _s$, for $\langle v_i\rangle _s$ the $i$ component of the bulk velocity of particle species $s$. 

A collisionless equilibrium DF is a solution of the steady-state Vlasov equation. A method frequently used to solve Vlasov's equation is to write $f_{s}$ as a function of a subset of the constants of motion (Jeans' Theorem) (see \citet{Schindlerbook} for example). This paper considers collisionless plasmas described by DFs of the form
\begin{equation}
f_{s}=\frac{n_0}{(\sqrt{2\upi}v_{th,s})^3}\; {\rm e}^{-\beta_sH_s}g_{s}(p_{xs},p_{ys};v_{th,s}),\label{eq:F_form}
\end{equation} 
with $g_{s}$ the unknown deviation from a Maxwellian distribution, parameterised by the thermal velocity $v_{th,s}$ of particle species $s$. This form is chosen for the DF for practical mathematical reasons (integrability) and to be readily compared to the Maxwellian distribution function when $g_{s}=1$. Note that for DFs of the form in equation (\ref{eq:F_form}), $\langle v_z\rangle_s=0$, since $f_{s}$ is an even function of $v_z$. The species-dependent parameter $\beta_{s}=1/(k_BT_s)$ is the thermal beta, with $n_{0}$ a normalisation parameter that does not necessarily represent the number density. The combination of quasineutrality ($N_{i}(A_{x},A_{y},\phi)=N_{e}(A_{x},A_{y},\phi)$) and a DF of the form in equation (\ref{eq:F_form}) results in a scalar potential that is implicitly defined as a function of the vector potential, e.g. \citet{Harrison-2009a,Schindlerbook,Tasso-2014,Kolotkov-2015}:
\begin{equation}
\phi_{qn}(A_x,A_y)=\frac{1}{e(\beta_e+\beta_i)}\ln (N_i/N_e),\label{eq:quasineutral}
\end{equation} 
where $N_i(A_x,A_y)$ and $N_e(A_x,A_y)$ are the number densities of the ions and electrons respectively, and $e$ is the elementary charge. In this work, parameters are chosen such that $N_i=N_e$ as functions over $(A_x,A_y)$ space, and so `strict neutrality' is satisfied, implying $\phi_{qn}=0$. It has been shown in \citet{Channell-1976} that this form of DF, together with strict neutrality, implies that the relevant component of the pressure tensor, $\mathsfbi{P}_{zz}$, is a 2-D integral transform of the unknown function $g_{s}$, given by
\begin{eqnarray}
&&\mathsfbi{P}_{zz}(A_x,A_y)=\frac{\beta_{e}+\beta_{i}}{\beta_{e}\beta_{i}}\frac{n_0}{2\upi m_{s}^2v_{th,s}^2}\nonumber\\
&&\times\int_{-\infty}^\infty\int_{-\infty}^\infty \; {\rm e}^{-\beta_{s}\left((p_{xs}-q_{s}A_x)^2+(p_{ys}-q_{s}A_y)^2\right)/(2m_{s})}g_{s}(p_{xs},p_{ys};v_{th,s})dp_{xs}dp_{ys}. \label{eq:Channell}
\end{eqnarray}
This equation defines the inverse problem at hand, viz. `for a given macroscopic equilibrium characterised by $\mathsfbi{P}_{zz}(A_x,A_y)$, can we invert the transform to solve for the unknown function $g_{s}$?' Note that the current densities 
\begin{eqnarray}
j_x(A_x,A_y)=\sum_{s}q_{s}n_s\langle v_x\rangle_s=\sum_{s}\,q_{s}\,\int v_{x} \,f_{s} \, d^3v,\nonumber\\
j_y(A_x,A_y)=\sum_{s}q_{s}n_s\langle v_y\rangle_s=\sum_{s}\,q_{s}\,\int v_{y} \, f_{s} \, d^3v,\nonumber
\end{eqnarray}
are themselves related to the pressure according to 
\begin{equation}
\boldsymbol{j}(A_x,A_y)=\frac{\partial \mathsfbi{P}_{zz}}{\partial \boldsymbol{A}},
\end{equation} 
see \citet{Grad-1961, Mynick-1979a,Schindlerbook,Harrison-2009a} for example. 

The above equation demonstrates that to reproduce a specific magnetic field, the $\mathsfbi{P}_{zz}$ function must be compatible. For example, in the case of a \emph{force-free field}, there is a simple procedure one can follow to calculate an expression for $\mathsfbi{P}_{zz}(A_x,A_y)$ (for details see Section 3).

In \citet{Abraham-Shrauner-1968}, Hermite polynomials are used to solve the Vlasov-Maxwell (VM) system for the case of `stationary waves' in a manner like that to be described in this paper. These correspond not to Vlasov equilibria, but rather to nonlinear waves that are stationary in the wave frame.

In \citet{Channell-1976}, two methods are presented for the solution of the inverse problem with neutral VM equilibria. These two methods are inversion by Fourier transforms and -- once again -- expansion over Hermite polynomials respectively. First impressions suggest that Fourier transforms do seem ideally suited to the task, since the right-hand side of equation (\ref{eq:Channell}) allows the convolution theorem to be applied. The Fourier transform method is used in \citet{Channell-1976} and \citet{Harrison-2009a} for example. However, when either the Fourier or inverse Fourier transform cannot be calculated, this method clearly fails to be of use. 

The method presented in this paper should be seen as a rigorous extension/generalisation of the Hermite Polynomial method used by Abraham-Shrauner and Channell. As such it is complementary to the Fourier transform method.  

The structure of this paper is as follows. Section 2 contains the mathematical details of the solution of the inverse problem defined in the Introduction. First, a formal solution is derived in Subsection \ref{sec:formal}, by using known methods of inverting Weierstrass transforms with possibly infinite series of Hermite polynomials. For the formal solution to meaningfully describe a DF however, these series must be convergent, positive and bounded. A sufficient condition for convergence that places a restriction on the pressure tensor is obtained in Subsection \ref{sec:convergence}. In Subsection \ref{sec:non-neg} we argue that for an appropriate pressure function, there always exists a positive DF, for a sufficiently magnetised plasma. We include some technical calculations in Appendix B that support the positivity argument, including proofs for a certain class of function.

In Section 3. we present non-trivial examples to demonstrate the application of the inversion method to  a recently derived force-free DF, \citet{Allanson-2015PoP}, as well as to DFs that correspond to the same magnetic field, but in a different gauge. This work is motivated by numerical reasons, and should allow easier calculation and visualisation of the DFs. In Appendix A we present the full details of the calculations that verify that these DFs satisfy the convergence criteria derived in Subsection \ref{sec:convergence}, and that as a result the DFs are bounded. In Section \ref{Sec:Channell} we consider the use of the method for a non-force-free magnetic field, considered by \citet{Channell-1976} using Fourier transforms. This calculation is included to demonstrate the relationship between the Fourier transform and Hermite Polynomial inversion methods.

\section{Solution of the inverse problem}
To make mathematical progress, we make the assumption of either `summative' or `multiplicative' separability, i.e. that $\mathsfbi{P}_{zz}(A_x,A_y)$ is of the form
\begin{equation}
\mathsfbi{P}_{zz}=\frac{n_0(\beta_e+\beta_i)}{\beta_e\beta_i}\left(\tilde{P}_1(A_x)+\tilde{P}_2(A_y)\right)\;{\rm or}\;\mathsfbi{P}_{zz}=\frac{n_0(\beta_e+\beta_i)}{\beta_e
\beta_i}\tilde{P}_1(A_x)\tilde{P}_2(A_y).\label{eq:pform}
\end{equation}
The components of the pressure, $\tilde{P}_1(A_x)$ and $\tilde{P}_2(A_y)$, are dimensionless. These assumptions are commensurate with 
\begin{equation}
g_{s}=g_{1s}(p_{xs};v_{th,s})+g_{2s}(p_{ys};v_{th,s})\;{\rm or}\;g_{s}=g_{1s}(p_{xs};v_{th,s})g_{2s}(p_{ys};v_{th,s})\label{eq:gsep},
\end{equation}
respectively, and allow separation of variables according to 
\begin{eqnarray}
\tilde{P}_1(A_x)=\frac{1}{\sqrt{2\upi} m_{s}v_{th,s}}\int_{-\infty}^{\infty}\; {\rm e}^{-\beta_{s}\left(p_{xs}-q_{s}A_x\right)^2/(2m_{s})}g_{1s}(p_{xs};v_{th,s})dp_{xs},\label{eq:p1tog1}\\ 
\tilde{P}_2(A_y)=\frac{1}{\sqrt{2\upi} m_{s}v_{th,s}}\int_{-\infty}^{\infty}\; {\rm e}^{-\beta_{s}\left(p_{ys}-q_{s}A_y\right)^2/(2m_{s})}g_{2s}(p_{ys};v_{th,s})dp_{ys}.\label{eq:p2tog2}
\end{eqnarray}
The separation constant is set to unity in the case of multiplicative separability, and zero in the case of additive separability, without loss of generality. The components of the pressure are now represented by 1-D integral transforms of the unknown parts of the DF. 

\subsection{Weierstrass transform}\label{sec:Weierstrass}
The Weierstrass transform, $\Phi(x)$ of $\phi(y)$, is defined by
\begin{equation}
\Phi(x)=\mathcal{W}\left[\phi\right]:x=\frac{1}{\sqrt{4\upi}}\int_{-\infty}^\infty \, {\rm e}^{-(x-y)^2/4}\,\phi(y)\,dy,\label{eq:Weier}
\end{equation}
see \citet{Bilodeau-1962} for example. This is also known as the Gauss transform, Gauss-Weiertrass transform and the Hille transform \citep{Widder-1951}. As the Green's function solution to the heat/diffusion equation, $\Phi(x)$ represents the temperature/density profile of an infinite rod one second after it was $\phi(x)$, see \citet{Widder-1951}, implying that the Weierstrass transform of a positive function is itself a positive function.
$\tilde{P}_1$ and $\tilde{P}_2$ are expressed as Weierstrass transforms of $g_{1s}$ and $g_{2s}$ in equations (\ref{eq:p1tog1}) and (\ref{eq:p2tog2}) respectively, give or take some constant factors. Formally, the operator for the inverse transform is $\; {\rm e}^{-D^2}$, with D the differential operator and the exponential suitably interpreted, see \citet{Eddington-1913, Widder-1954} for two different interpretations of this operator. We should mention that one of the existing nonlinear force-free VM equilibria known \citep{Harrison-2009b} is based on an eigenfunction of the Weierstrass transform \citep{Wolf-1977}. 

Perhaps a more computationally `practical' method employs Hermite polynomials, see \citet{Bilodeau-1962}.
The Weierstrass transform of the $n^{th}$ Hermite polynomial $H_n(y/2)$ is $x^n$. Hence if one knows the coefficients of the Maclaurin expansion of $\Phi(x)$ in equation (\ref{eq:Weier}), 
\[
\Phi(x)=\sum_{j=0}^\infty\eta_jx^j,
\]
then the Weierstrass transform can immediately be inverted to obtain the formal expansion
\begin{equation}
\phi(y)=\sum_{j=0}^\infty\eta_jH_j\left(y/2\right).\label{eq:inverseweier}
\end{equation}
For this method to be useful in our problem, the pressure function must have a Maclaurin expansion that is convergent over all $(A_x,A_y)$ space. Then, its coefficients of expansion must `allow' the Hermite series to converge. Questions regarding the positivity and convergence of formal solutions represented by infinite series of Hermite polynomials were raised by \citet{Abraham-Shrauner-1968} and \citet{Hewett-1976} respectively, and the same questions arise in the context of the problems in this paper. For some other examples of applications of Hermite polynomials to collisionless and weakly collisional plasmas, see \citet{Camporeale-2006, Suzuki-2008, Zocco-2015, Schekochihin-2016}. We also remark that the use of Hermite polynomials in kinetic theory dates back, at least, to \citet{Grad-1949a,Grad-1949b} in the study of rarefied collisional gases. 
\subsection{Formal solution}\label{sec:formal}
The following discussion applies to pressure functions of both summative and multiplicative form, with Maclaurin expansion representations (convergent over all $(A_x,A_y)$ space) given by
\begin{equation}
\tilde{P}_1(A_x)=\sum_{m=0}^\infty a_m\left(\frac{A_x}{B_0L}\right)^m,\;\;\tilde{P}_2(A_y)=\sum_{n=0}^\infty b_n\left(\frac{A_y}{B_0L}\right)^n,\label{eq:Pmacro}
\end{equation}
with $B_0$ and $L$ the characteristic magnetic field strength and spatial scale respectively. In line with the discussion on inversion of the Weierstrass transform in Subsection \ref{sec:Weierstrass}, we solve for $g_{s}$ functions represented by the following expansions
\begin{eqnarray}
g_{1s}(p_{xs};v_{th,s})&=&\displaystyle\sum_{m=0}^\infty C_{ms}H_m\left(\frac{p_{xs}}{\sqrt{2}m_{s}v_{th,s}}\right),\label{eq:gform1}\\
g_{2s}(p_{ys};v_{th,s})&=&\displaystyle\sum_{n=0}^\infty D_{ns}H_n\left(\frac{p_{ys}}{\sqrt{2}m_{s}v_{th,s}}\right),\label{eq:gform2}
\end{eqnarray}
with currently unknown species-dependent coefficients $C_{ms}$ and $D_{ns}$. We cannot simply `read off' the coefficients of expansion as in (\ref{eq:inverseweier}), since our integral equations are not quite in the `perfect form' of (\ref{eq:Weier}). Upon computing the integrals of equations (\ref{eq:p1tog1}) and (\ref{eq:p2tog2}) with the above expansions for $g_{s}$, we have
\begin{equation}
\tilde{P}_1(A_x)=\displaystyle\sum_{m=0}^\infty \left(\frac{\sqrt{2}q_{s}}{m_{s}v_{th,s}}\right)^{m}C_{ms}\,A_x^{m},\;\;\tilde{P}_2(A_y)=\displaystyle\sum_{n=0}^\infty \left(\frac{\sqrt{2}q_{s}}{m_{s}v_{th,s}}\right)^{n}D_{ns}\,A_y^n.\label{eq:Pmicro}
\end{equation}
This result appears species dependent. However, to ensure neutrality ($N_i(A_x,A_y)=N_e(A_x,A_y)$) - as in \citet{Channell-1976, Harrison-2009b, Wilson-2011} - we have to fix the pressure function to be species independent. It clearly must also match with the pressure function that maintains equilibrium with the prescribed magnetic field. The conditions derived here are critical for making a link between the macroscopic description of the equilibrium structure  with the microscopic one of particles. These requirements imply by the matching of equations (\ref{eq:Pmacro}) and (\ref{eq:Pmicro}) that
\begin{eqnarray}
\left(\frac{\sqrt{2}q_{s}}{m_{s}v_{th,s}}\right)^{m}C_{ms} & = & \left(\frac{1}{B_0L}\right)^{m}a_{m}\implies C_{ms}={\rm sgn}(q_{s})^m\left(\frac{\delta_s}{\sqrt{2}}\right)^{m}a_{m}, \label{eq:ccoeff}\\
\left(\frac{\sqrt{2}q_{s}}{m_{s}v_{th,s}}\right)^{n}D_{ns} & =  &\left(\frac{1}{B_0L}\right)^nb_n\implies D_{ns}={\rm sgn}(q_{s})^n\left(\frac{\delta_s}{\sqrt{2}}\right)^{n}b_{n}\label{eq:dcoeff},
\end{eqnarray}
with ${\rm sgn}(q_e)=-1$ and ${\rm sgn}(q_i)=1$. The species-dependent magnetisation parameter, $\delta_s$, see \citet{Fitzpatrickbook} for example, is defined by 
\[
\delta_s=\frac{m_{s}v_{th,s}}{eB_0L}.
\]
It is the ratio of the thermal Larmor radius, $\rho_s=v_{th,s}/|\Omega_s|$, to the characteristic length scale of the system, $L$. The gyrofrequency of particle species $s$ is $\Omega_s=q_{s}B_0/m_{s}$. The magnetisation parameter is also known as the fundamental ordering parameter in gyrokinetic theory (see \citet{Howes-2006} for example). (In particle orbit theory, $\delta_s\ll 1$ implies that a guiding centre approximation will be applicable for that species, see \citet{Northrop-1961}.)

\subsection{Convergence of the distribution function}\label{sec:convergence}
Here we find a sufficient condition that, when satisfied, guarantees that the Hermite series representations in (\ref{eq:gform1}) and (\ref{eq:gform2}) converge. This provides some answers to questions on the convergence of Hermite Polynomial representations of Vlasov equilibria dating back to \citet{Hewett-1976}. 
%Note for the following discussion, that when we write \[\mathcal{J}_n\sim o\left(\mathcal{K}_n\right),\] we mean that \[ \mathcal{J}_n/\mathcal{K}_n\to 0\;{\rm as}\;n\to\infty.\]
\begin{theorem}\label{thm:HermiteConvergence}
Consider a Maclaurin expansion of the form
\begin{equation}
\tilde{P}_j(A)=\sum_{m=0}^\infty a_m\left(\frac{A}{B_0L}\right)^m
\end{equation}
that is convergent for all $A$. Then for $\varepsilon_s=m_{s}^2v_{th,s}^2/2$ the function $g_{js}$, calculated in the inverse problem defined by the association
\begin{equation}
\tilde{P}_{j}(A):=\tilde{P}_{INT}(A)=\frac{1}{\sqrt{4\upi\varepsilon_s}}\int_{-\infty}^{\infty} \; {\rm e}^{-(p_{s}-q_{s}A)^2/(4\varepsilon_s)}g_{js}(p_{s};v_{th,s})dp_{s}.\label{eq:sample}
\end{equation}
of the form
\begin{equation}
g_{js}(p_s;v_{th,s})=\sum_{m=0}^\infty a_m\,{\rm sgn}(q_{s})^m\left(\frac{\delta_s}{\sqrt{2}}\right)^mH_{m}\left(\frac{p_{s}}{\sqrt{2}m_{s}v_{th,s}}\right)\label{eq:gj}
\end{equation}
converges for all $p_{s}$, provided 
\begin{equation}
\lim_{m\to\infty}\sqrt{m}\left|\frac{a_{m+1}}{a_m}\right|<1/\delta_s,
\end{equation}
in the case of a series composed of both even- and odd-order terms, or
\begin{equation}
\lim_{m\to\infty}\,m\,\left|\frac{a_{2m+2}}{a_{2m}}\right|<1/(2\delta_s^2),\hspace{3mm}\lim_{m\to\infty}\,m\,\left|\frac{a_{2m+3}}{a_{2m+1}}\right|<1/(2\delta_s^2),
\end{equation}
in the case of a series composed only of even-, or odd-order terms, respectively.
\end{theorem}
\begin{proof}
For a series composed of even- and odd-order terms, we have that
\begin{equation}
g_{s}(p_{s};v_{th,s})=\sum_{m=0}^{\infty}a_{m}\,{\rm sgn}(q_{s})^m\left(\frac{\delta_s}{\sqrt{2}}\right)^mH_{m}\left(\frac{p_{s}}{\sqrt{2}m_{s}v_{th,s}}\right)\label{eq:g2}.
\end{equation}
An upper bound on Hermite polynomials (see e.g.  \citet{Sansonebook}) is provided by the identity 
%There is an identity in \citep{Sansonebook} that provides an upper bound on Hermite polynomials  
\begin{equation}
| H_{j}(x)|<k\sqrt{j!}2^{j/2}\exp\left(x^2/2\right)\; \mbox{\rm s.t.\ } \; k=1.086435\, .\label{eq:hermbound}
\end{equation}
This upper bound implies
\[
a_m\,{\rm sgn}(q_{s})^m\left(\frac{\delta_s}{\sqrt{2}}\right)^mH_m\left(\frac{p_{s}}{\sqrt{2}m_{s}v_{th,s}}\right)<ka_m\delta_s ^m\sqrt{m!}\exp\left(\frac{p_{s}^2}{4m_{s}^2v_{th,s}^2}\right).
\]
Working on the level of the series composed of upper bounds, the ratio test clearly requires
\begin{eqnarray}
\lim_{m\to\infty}\Bigg|\frac{a_{m+1}}{a_m}\Bigg|\sqrt{m+1}<1/\delta_s,\nonumber\\
\implies \lim_{m\to\infty}\Bigg|\frac{a_{m+1}}{a_m}\Bigg|\sqrt{m}<1/\delta_s, \label{eq:criterion}
\end{eqnarray}
for a given $\delta_s\in (0,\infty)$. Then, the comparison/squeeze test implies that if the condition of equation (\ref{eq:criterion}) is satisfied, that since the series composed of upper bounds will converge, so must $g_{s}(p_{s})$.  An analogous argument holds for those series with only even or odd order terms, with the ratio test giving
\begin{equation}
\lim_{m\to\infty}\Bigg|\frac{a_{2m+2}}{a_{2m}}\Bigg|m<1/(2\delta_s^2), \hspace{2mm}\text{or}\hspace{2mm}\lim_{m\to\infty}\Bigg|\frac{a_{2m+3}}{a_{2m+1}}\Bigg|m<1/(2\delta_s^2), \label{eq:criterion2}
\end{equation}
respectively. By the same argument as above, the comparison test implies that if the condition of (\ref{eq:criterion2}) is satisfied, that since the series composed of upper bounds will converge, so must $g_{s}(p_{s})$.
%\qedhere
\end{proof}

\subsection{Positivity of the distribution function}\label{sec:non-neg}
In this Subsection, we consider the positivity of the Hermite series representation of $g_{s}$ -- given by equations (\ref{eq:gform1}) and (\ref{eq:gform2}) -- and hence  positivity of the DF. This provides some answers to questions on the positivity of DF representation by Hermite polynomials dating back to \citet{Abraham-Shrauner-1968}, and also raised by \citet{Hewett-1976}. 

For an example of a $g_s$ function that is not necessarily always positive despite the pressure function being positive, consider a pressure function (e.g. from \citet{Channell-1976}) that is quadratic in the vector potential. In our notation, the pressure function considered by Channell is
\[
\tilde{P}=\frac{1}{2}\left(a_0+a_2\left(\frac{A_x}{B_0L}\right)^2\right)+\frac{1}{2}\left(a_0+a_2\left(\frac{A_y}{B_0L}\right)^2\right).
\]
The resultant $g_{s}$ function is of the form
\[
g_{s}\propto \frac{1}{2}\left[a_0+a_2\left(\frac{\delta_s}{\sqrt{2}}\right)^2H_2\left(\frac{p_{xs}}{\sqrt{2}m_{s}v_{th,s}}\right)\right]+\frac{1}{2}\left[a_0+a_2\left(\frac{\delta_s}{\sqrt{2}}\right)^2H_2\left(\frac{p_{ys}}{\sqrt{2}m_{s}v_{th,s}}\right)\right].
\]
Once these Hermite polynomials are expanded, by substituting $p_{xs}=p_{ys}=0$ we see that positivity of $g_s$ is -- for given values of $a_0$ and $a_2$ -- contingent on the size of $\delta_s$, 
\[a_0-a_2\delta_{s}^{2}>0\implies \delta_s^2<\frac{a_0}{a_2}.\]
However, there is not necessarily anything `special' about the point $0$, as compared to other points in momentum-space. For example, consideration of the pressure function
\[
\tilde{P}_j=\left(a_0+a_2\left(\frac{A}{B_0L}\right)^2 +a_4\left(\frac{A}{B_0L}\right)^4       \right),
\]
gives a $g_s$ function that can, for given values of $a_0,a_2,a_4$ and for $\delta_s$ sufficiently large, be positive at $p_s=0$, and negative at some other points. 

It is worth considering how a $g_{s}$ function that is negative for some $p_{s}$ can transform in the manner of (\ref{eq:p1tog1}) and (\ref{eq:p2tog2}) to give a positive $\tilde{P}_j(A)$. One might expect that for certain values of $A$ such that the Gaussian 
\[
{\rm e}^{-(p_{s}-q_{s}A)^2/(4\varepsilon_s)}
\]
is centred on the region in $p_s$ space for which $g_s$ is negative, that a negative value of $\tilde{P}_j(A)$ could be the result.

Essentially, the Gaussian will only `successfully sample' a negative region of $g_s$ to give a negative value of $\tilde{P}_j(A)$ if the Gaussian is narrow enough -- for a given value of $\varepsilon_s$ --  to `resolve' a negative patch of $g_s$. In other words, if the Gaussian is too broad, it won't `see' the negative patches of $g_s$, and hence $\tilde{P}_j(A)$ will be positive. Hence the non-negativity of $\tilde{P}_j(A)$ is a restriction on the possible shape of $g_s$, and how that shape must scale with $\varepsilon_s$.

It is a short algebraic exercise to rewrite (\ref{eq:sample}) in the form
\begin{equation}
\sum_{n=0}^\infty a_n\left(\text{sgn}(q_s)\delta_s\tilde{A}\right)^n=\frac{1}{\sqrt{2\pi}}\int_{-\infty}^\infty e^{-(\tilde{p}_s-\tilde{A})^2/2}\bar{g}_s(\tilde{p}_s;\delta_s)d\tilde{p}_s,\label{eq:newsample}
\end{equation}
by using the following associations
\[
\tilde{A}=\frac{A}{B_0L},\hspace{3mm}\tilde{p}_s=\frac{p_s}{\sqrt{2\varepsilon_s}},\hspace{3mm}g_s(p_s;\varepsilon_s)=\bar{g}_s(\tilde{p}_s;\delta_s),
\]
and with
\begin{equation}
\bar{g}_s(\tilde{p}_s;\delta_s)=\sum_{n=0}^\infty a_n\text{sgn}(q_s)^n\left(\frac{\delta_s}{\sqrt{2}}\right)^nH_n\left(\frac{\tilde{p}_s}{\sqrt{2}}\right).\label{eq:gnorm}
\end{equation}
We shall assume that the right-hand side of (\ref{eq:gnorm}) represents a differentiable function. Note that the Gaussian in (\ref{eq:newsample}) is of fixed width $2\sqrt{2}$ (defined at $1/e$), in contrast to the Gaussian of variable width defined in (\ref{eq:sample}). 

If the Hermite series satisfies the condition in Theorem 1 then it is convergent, so (\ref{eq:hermbound}) gives
\begin{equation}
\left|\bar{g}_s(\tilde{p}_s;\delta_s)\right|<Le^{\tilde{p}_s^2/4}      
\nonumber
\end{equation}
for some finite and positive $L$, determined by the sum of the (possibly infinite) series. Note that these bounds automatically imply integrability of $f_s$ since, for some finite $L^\prime>0$, we have that $\left|\bar{g}_s(\tilde{p}_s;\delta_s)\right|<L^\prime e^{\tilde{p}_s^2/2}$  implies integrability, which is a less strict condition. This can be seen from (\ref{eq:newsample}).

The bounds on $\bar{g}_s$ given above demonstrate that $\bar{g}_s$ can not tend to infinity for finite $\tilde{p}_s$. Hence it can only reach $-\infty$ as $|\tilde{p}_s|\to\infty$.  We argue however that the positivity of the pressure prevents the possibility of $\bar{g}_s$ being without a finite lower bound. The heuristic reasoning is as follows: the expression on the right-hand side of (\ref{eq:newsample}) treats -- in the language of the heat/diffusion equation -- the $\bar{g}_{s}$ function as the initial condition for a temperature/density distribution on an infinite 1-D line, and the left-hand side represents the distribution at some finite time later on (half a second later, see \citet{Widder-1951}). Were $\bar{g}_s$ to be unbounded from below, this would imply for our problem that a smooth `temperature/density' distribution that is initially unbounded from below could, in some finite time, evolve into a distribution that has a positive and finite lower bound. This seems entirely unphysical since this would imply that an infinite negative `sink' of heat/mass would somehow be `filled in' above zero level in a finite time. %Note that one might wish to raise the example of a negative Dirac delta here as a potential counter example, since Dirac delta functions are initially unbounded and evolve in finite time to be a bounded Gaussian. However, Dirac delta functions are not differentiable and hence we can not consider these as initial conditions. One could however consider a function partly formed of an infinite sequence of negative and ever-more narrow `unit-mass' Gaussians, that approach a Dirac delta function as $|\tilde{p}_{s}|\to\infty$. This would be a smooth initial condition that is unbounded from below, however it will remain unbounded from below in finite time. 
In Appendix \ref{App:nonneg} we give some more technical mathematical arguments to support our claim that this is not possible, including proofs for a certain class of $\bar{g}_s$ functions.

If $\bar{g}_s$ (and hence $g_{s}$) is indeed bounded below then that means that one can always add a finite constant to $g_s$ to make it positive, should the lower bound be known. However this constant contribution would directly correspond to raising the pressure (through the zeroth order Maclaurin coefficient $a_0$). But if we wish to consider a pressure function that is `fixed', then we have a fixed $a_0$, and so it is not immediately obvious whether or not we can obtain a $g_s$ that is positive over all momentum space. We have already seen some examples in the discussion above for which the sign of $g_s$ depended on the value of $\delta_s$. Consider $\bar{g}_s$ evaluated at some particular value of $\tilde{p}_s$. We see from (\ref{eq:gnorm}) that positivity requires
\[
a_0+c_1\delta_s+c_2\delta_s^2+...>0,
\]
for $c_1,c_2,...$ finite constants. We also know that $a_0>0$ since $P(0)>0$, i.e. the pressure is positive. This clearly demonstrates that positivity of $g_s$ places some restriction on possible values of $\delta_s$.

Let us now suppose that for a given value of $\delta_s$, that there exists some regions in $\tilde{p}_s$ space where $\bar{g}_s<0$. Our claim that $\bar{g}_s$ has a finite lower bound, combined with the expression in (\ref{eq:gnorm}) implies that the $\bar{g}_s$ function is bounded below by a finite constant of the form $a_0+\delta_s \mathcal{M}$, with
\[
\mathcal{M}=\frac{1}{\sqrt{2}}\inf_{\tilde{p}_s}\sum_{n=1}^\infty a_n\text{sgn}(q_s)^n\left(\frac{\delta_s}{\sqrt{2}}\right)^{n-1}H_n\left(\frac{\tilde{p}_s}{\sqrt{2}}\right),
\]
and finite. By letting $\delta_s\to 0$ we see that $\bar{g}_s$ will converge uniformly to $a_0$, with
\[
\lim_{\delta_s\to 0}\bar{g}_s(\tilde{p}_s,\delta_s)=a_0>0.
\]
Hence, there must have existed some critical value of $\delta_s=\delta_c$ such that for all $\delta_s<\delta_c$ we have positivity of $\bar{g}_s$. Note that if the negative patches of $\bar{g}_s$ do not exist for any $\delta_s$, then trivially $\delta_c=\infty$ as a special case.\\

To summarise, we claim  -- provided $g_s$ is differentiable and convergent -- that for values of the magnetisation parameter $\delta_s$ less than some critical value $\delta_c$, according to $0<\delta_s<\delta_c\le \infty$, $g_s$ is positive for any positive pressure function.

\section{Examples: DFs for nonlinear force-free magnetic fields }
\subsection{Basic theory of 1-D force-free fields}
Force-free fields are those whose current density is everywhere parallel to the magnetic field, giving zero Lorentz force
\begin{equation}
\boldsymbol{j}=\alpha\boldsymbol{B}\iff \boldsymbol{j}\times\boldsymbol{B}=\boldsymbol{0} .
\end{equation}
The nature of $\alpha$ determines three distinct classes. Potential fields have $\alpha=0$, linear force-free fields have $\alpha={\rm const.}$ and nonlinear force-free fields have $\alpha=\alpha(\boldsymbol{r})$. One-dimensional force-free fields can be represented without loss of generality by 
\begin{equation}
\boldsymbol{B}=\left(B_x(z),B_y(z),0\right)=\left(-\frac{dA_y}{dz},\frac{dA_x}{dz},0\right),\;B^2={\rm const.}
\end{equation}
This leads on to a pressure balance of the form
\begin{equation}
\frac{d}{dz}\mathsfbi{P}_{zz}=0\implies \mathsfbi{P}_{zz}={\rm const.}
\end{equation}
As demonstrated in \citet{Harrison-2009b, Neukirch-2009}, the assumption of summative separability (the first option in equation (\ref{eq:pform})) determines the components of the pressure according to 
\begin{equation}
n_0\frac{\beta_e+\beta_i}{\beta_e\beta_i}\tilde{P}_1(A_x)+\frac{1}{2\mu_0}B_y^2(A_x)={\rm const.},\;\;n_0\frac{\beta_e+\beta_i}{\beta_e\beta_i}\tilde{P}_2(A_y)+\frac{1}{2\mu_0}B_x^2(A_y)={\rm const.}
\end{equation}
These expressions can now be used as the left-hand side of the integral equations (\ref{eq:p1tog1}) and (\ref{eq:p2tog2}), and one could attempt to invert the Weierstrass transforms. This method was used in \citet{Harrison-2009b} to derive a summative pressure for the `force-free Harris Sheet' (FFHS) magnetic field, and derive the corresponding DF.

As shown in \citet{Harrison-2009a},  Amp\`{e}re's law admits an infinite number of pressure functions for the same force-free equilibrium. Once a $\mathsfbi{P}_{zz}(A_x,A_y)$ with the correct properties has been found one can define another pressure function giving rise to the same current density by using the nonlinear transformation
\begin{equation}
\bar{\mathsfbi{P}}_{zz}(A_x,A_y)=\psi^\prime(P_{ff})^{-1}\psi(\mathsfbi{P}_{zz})\label{eq:Ptrans}.
\end{equation} 
Here any differentiable, non-constant function $\psi$ can be used, such that the right-hand side is positive, with $P_{ff}$ the pressure, $\mathsfbi{P}_{zz}$, evaluated at the force-free vector potential $\boldsymbol{A}_{ff}$. 

Obviously, even if the integral equation (\ref{eq:Channell}) can be solved for the original function $\mathsfbi{P}_{zz}(A_x, A_y)$ it is by no means
clear that this is possible for the transformed function $\bar{\mathsfbi{P}}_{zz}$. Usually one would expect that solving (\ref{eq:Channell}) for $g_{s}$ is much more difficult after the transformation to $\bar{\mathsfbi{P}}_{zz}$. This pressure transformation theory is important for the derivation of the low-beta DF for the nonlinear FFHS \citep{Allanson-2015PoP}. As explained therein, if the pressure transformation
\begin{equation}
\psi(\mathsfbi{P}_{zz})=\exp\left[\frac{1}{P_0}\left(\mathsfbi{P}_{zz}-P_{ff}\right)\right],\label{eq:Pfunc}
\end{equation}
is used, for $P_0$ a positive constant, it can be readily seen that $\bar{\mathsfbi{P}}_{zz}|_{\boldsymbol{A}_{ff}}=P_0$ and so free manipulation of the constant pressure is possible. This is of particular interest because it allows us to freely choose the plasma beta, $\beta_{pl}$, the ratio between
the thermal and magnetic energy densities (in our system the gas/plasma pressure and the magnetic pressure respectively) 
\[
\beta_{pl}=\frac{k_B}{(B_0^2/2\mu_0)}\sum_{s}n_sT_s=\frac{2\mu_0\mathsfbi{P}_{zz}}{B_0^2}.
\] 

\subsection{On the gauge for the vector potential}\label{Subsec:regauge}
A free choice of the plasma beta is not possible in the summative \emph{Harrison-Neukirch} equilibrium DF since that equilibrium has a lower bound of unity for the plasma beta. Note that the $\mathsfbi{P}_{zz}$ used in that work is of a `summative form' 
\[\mathsfbi{P}_{zz}=P_1(A_x)+P_2(A_y).\label{eq:oldgauge}\]
In fact it seems to be a feature generally observed that for pressure tensors (that correspond to force-free fields) constructed in this manner (\citet{Harrison-2009b, Abraham-Shrauner-2013, Wilson-2011, Kolotkov-2015}) the plasma-beta is necessarily bounded below by unity. A recent paper, \citet{Allanson-2015PoP}, used the pressure transformation techniques described above, resulting in a pressure tensor of `multiplicative form'
\[\mathsfbi{P}_{zz}=P_1(A_x)P_2(A_y),\]
to construct a DF with any $\beta_{pl}$. However, the exact form of the DF was challenging to calculate numerically for low $\beta_{pl}$, with plots for $\beta_{pl}$ only modestly below unity presented ($\beta_{pl}=0.85$). The `problem terms' are those that depend on $p_{xs}$. The specific problem is that the $A_x$ function used in previous papers is neither even nor odd as a function of $z$, 
\[A_x=2B_0L\arctan \left(\exp\left(\frac{z}{L}\right)\right),\]
and as a result the range of $p_{xs}$ for which it is necessary to numerically calculate a convergent DF can be obstructive, say over a symmetric range in velocity space. Specifically, it is challenging to attain numerical convergence for sums over Hermite polynomials when the modulus of the argument is large. When $A_x$ is neither even nor odd, then $|p_{xs}|$ can take on larger than `necessary' values for a given $v_x$.

Hence, in this paper, we shall `re-gauge' the vector potential component $A_x$ to be an odd function, 
\begin{equation}A_x=2B_0L\arctan\left(\tanh\left(\frac{z}{2L}\right)\right),\label{eq:gauge}\end{equation}
which is commensurate with $B_y$ being an even function and results in the same $B_y=B_0\,{\rm sech}(z/L)$ as the one derived from the $A_x$ defined in (\ref{eq:oldgauge}). As a consequence the numerical calculation of the DFs that we shall calculate for the FFHS become easier in the low $\beta_{pl}$ regime. 

The structure of this section is as follows. In Subsection \ref{subsec:ANWT} we include the particulars of the recently derived FFHS equilibrium, in the original gauge, for completeness. In Subsection \ref{subsubsec:ANWTregaugemult} we calculate DFs corresponding to the \emph{`re-gauged' FFHS}, that are multiplicative. These `re-gauged' DFs are essentially equivalent to those derived in \citet{Allanson-2015PoP}, as functions of $z$ and $\boldsymbol{v}$. However they are different as functions of $\boldsymbol{p}_{s}$. The involved calculations that prove the necessary properties of convergence and boundedness of the above DFs, by using techniques established in this paper, are included in Appendix A.

\subsection{Multiplicative DF for the FFHS in the `original' gauge: $\beta_{pl}\in (0\,,\,\infty)$ }\label{subsec:ANWT}
The `summative' pressure used in \citet{Harrison-2009b} for a FFHS equilibrium is of the form
\begin{equation}
\mathsfbi{P}_{zz}(A_x,A_y)=\frac{B_0^2}{2\mu_0}\left[\frac{1}{2}\cos\left(\frac{2A_x}{B_0L}\right)+\exp\left(\frac{2A_y}{B_0L}\right)\right]+P_b.\label{eq:Pharrison}
\end{equation}
$P_b>B_0^2/(4\mu_0)$ is a constant that ensures positivity of $\mathsfbi{P}_{zz}$. This is the function that we exponentiate according to (\ref{eq:Ptrans}) and (\ref{eq:Pfunc}). To suit the problem we choose a pressure function and $g_{s}$ function of the form 
\begin{eqnarray*}
\bar{\mathsfbi{P}}_{zz}&=&n_0\exp\left(-\frac{1}{2\beta_{pl}}\right)\frac{\beta_e+\beta_i}{\beta_e\beta_i}\bar{P}_1(A_x)\bar{P}_2(A_y),\\
g_{s}&=&\exp\left(-\frac{1}{2\beta_{pl}}\right)g_{1s}(p_{xs};v_{th,s})g_{2s}(p_{ys};v_{th,s}).
\end{eqnarray*}

To use the method presented in Section 2., we now need to Maclaurin expand the complicated pressure function $\bar{P}_{zz}$. There is a result from combinatorics due to Eric Temple Bell that allows one to extract the coefficients of a power series, $f(x)$, that is itself the exponential of a known power series, $h(x)$, see \citet{Bell-1934}. If $f(x)$ and $h(x)$ are defined 
\begin{equation}
f(x)=\; {\rm e}^{h(x)}, \hspace{5mm} h(x)=\sum_{m=1}^{\infty}\frac{1}{m!}\zeta_mx^m ,\label{eq:exppower}
\end{equation}
then we can use `Complete Bell polynomials', also known as `Exponential Bell polynomials' and hereafter referred to as CBPs, to write $f(x)$ as
\begin{equation}
f(x)=\sum_{m=0}^\infty \frac{1}{m!}Y_m(\zeta_1,\zeta_2,...,\zeta_m)x^m. \label{eq:Bell}
\end{equation}
$Y_m(\zeta_1,\zeta_2,...\zeta_m)$ is the $m^{\rm{th}}$ CBP. Instructive references on CBPs can be found in \citet{Riordan-1958, Comtet-1974, Kolbig-1994, Connon-2010} for example. Here, the Maclaurin coefficients for the exponential and cosine functions of equation (\ref{eq:Pharrison}) are used as the arguments of the CBPs. These CBPs are used to form the Maclaurin coefficients of $\bar{P}_1$ and $\bar{P}_2$ as in equation (\ref{eq:Bell}). As detailed in \citet{Allanson-2015PoP}, the result is a pressure function of the form
\begin{equation}
\bar{\mathsfbi{P}}_{zz}=n_0\exp\left(\frac{-1}{2\beta_{pl}}\right)\frac{\beta_e+\beta_i}{\beta_e\beta_i}\sum_{m=0}^\infty a_{2m}\left(\frac{A_x}{B_0L}\right)^{2m}\sum_{n=0}^\infty b_n\left(\frac{A_y}{B_0L}\right)^{n} , \label{eq:Pmacl}
\end{equation} 
with $a_{2m}$ and $b_n$ defined by
\begin{eqnarray}
a_{2m}&=&\exp\left(\frac{1}{2\beta_{pl}}\right)\frac{(-1)^m2^{2m}}{(2m)!}Y_{2m}\left(0,\frac{1}{2\beta_{pl}}, 0 , ... , 0 , \frac{1}{2\beta_{pl}}\right),\label{eq:a2msimple}\\
b_n&=&\exp\left(\frac{1}{\beta_{pl}}\right)\frac{2^n}{n!}Y_n\left(\frac{1}{\beta_{pl}}, ..., \frac{1}{\beta_{pl}}\right).\label{eq:bnsimple}
\end{eqnarray}
The resultant DF is given
\begin{eqnarray}
&&f_{s}=\frac{n_{0}   e^{-1/(2\beta_{pl})}}{(\sqrt{2\upi}v_{th,s})^3}\; {\rm e}^{-\beta_sH_s}\nonumber\\
&&\times\sum_{m=0}^\infty a_{2m}\left(\frac{\delta_s}{\sqrt{2}}\right)^{2m}H_{2m}\left(\frac{p_{xs}}{\sqrt{2}m_{s}v_{th,s}}\right)\nonumber\\\
&&\times\sum_{n=0}^\infty b_{n}{\rm sgn}(q_{s})^n\left(\frac{\delta_s}{\sqrt{2}}\right)^{n}H_{n}\left(\frac{p_{ys}}{\sqrt{2}m_{s}v_{th,s}}\right).\label{eq:result}
\end{eqnarray}

 \subsection{Multiplicative DF for the `re-gauged' FFHS: $\beta_{pl}\in (0\,,\,\infty)$ }\label{subsubsec:ANWTregaugemult}
We will now calculate a multiplicative DF for the `re-gauged' FFHS, in the same style as \citet{Allanson-2015PoP}, in the effort to produce a low-beta DF for the FFHS that is easier to calculate numerically, and hence plot. This re-gauging is equivalent to adding a constant to $A_{x}$ and so corresponds to a shift in the origin of the $A_x$ dependent part of the summative $\mathsfbi{P}_{zz}$ used in \citet{Harrison-2009b}. As a result, one can derive a new summative pressure function in the same manner as in \citet{Harrison-2009b}, corresponding to this new gauge, as 
\begin{equation}
\mathsfbi{P}_{zz}=\frac{B_0^2}{2\mu_0}\left[\sin ^2\left(\frac{A_x}{B_0L}\right) + \exp\left(\frac{2A_y}{B_0L}\right)   \right]\label{eq:pnewgauge}
\end{equation}
The next step is to construct a multiplicative pressure tensor. Using the same pressure transformation technique as in \citet{Allanson-2015PoP} and Subsection \ref{subsec:ANWT}, on the $\mathsfbi{P}_{zz}$ given in equation (\ref{eq:pnewgauge}), we arrive at the `re-gauged' multiplicative pressure
\begin{eqnarray}\mathsfbi{P}_{zz}&=&P_0e^{-1/\beta_{pl}}\exp\left[ \frac{1}{\beta_{pl}}\left( \sin ^{2}\left(\frac{A_x}{B_0L}\right)+\exp\left(\frac{2A_y}{B_0L}\right)   \right)         \right]  \\
&=&P_0\exp{ \left[ \sum_{n=1}^\infty\frac{1}{(2n)!} \nu_{2n}\left(\frac{A_x}{B_0L}\right)^{2n}   \right]                }  \exp{ \left[\sum_{n=1}^\infty \frac{1}{n!} \xi_n \left(\frac{A_y}{B_0L}\right)^n  \right]                }     ,\end{eqnarray} 
with the coefficients defined by
\[ \nu_{2n}=\frac{(-1)^{n+1}2^{2n-1}}{\beta_{pl}}  ,\hspace{5mm}\xi_n=\frac{2^n}{\beta_{pl}}\, .\]
We now use the theory of CBPs, as in \citet{Allanson-2015PoP} and Subsection \ref{subsec:ANWT}, to write the pressure as
\begin{eqnarray}
\mathsfbi{P}_{zz}&=&P_0\sum_{m=0}^\infty \frac{1}{(2m)!}Y_{2m}\left(0\,,\,\nu_2\,,\,0\,,\, \nu_4\, ,\, ...\,,0\,,\, \nu_{2m}\right)\left(\frac{A_x}{B_0L}\right)^{2m}\nonumber\\
&\times&\sum_{n=0}^\infty \frac{1}{n!}Y_{n}\left(\xi_1\,,\,\xi_2\,,\, ...\,,\,\xi_n        \right)\left(\frac{A_y}{B_0L}\right)^{n}.\nonumber
\end{eqnarray}
Using a simple scaling argument as in \citet{Bell-1934,Connon-2010}, $Y_j(ax_1,a^2x_x,...,a^jx_j)=a^jY_j(x_1,x_2,...,x_j)$, gives
\begin{eqnarray}
\mathsfbi{P}_{zz}&=&P_0\sum_{m=0}^\infty \frac{(-1)^m2^{2m}}{(2m)!}Y_{2m}\left(0\,\,\frac{-1}{2\beta_{pl}}\,,\,0\,,\, \frac{-1}{2\beta_{pl}}\, ,\, ...\,,\,0\,,\, \frac{-1}{2\beta_{pl}}\right)\left(\frac{A_x}{B_0L}\right)^{2m}\nonumber\\
&\times&\sum_{n=0}^\infty \frac{2^m}{n!}Y_{n}\left(\frac{1}{\beta_{pl}}\,,\,\frac{1}{\beta_{pl}}\,,\, ...\,,\,\frac{1}{\beta_{pl}}          \right)\left(\frac{A_y}{B_0L}\right)^{n}.\nonumber
\end{eqnarray}
Using the methods established in this paper, namely expansion over Hermite polynomials, we calculate a DF that gives the above pressure
\begin{eqnarray}
&&f_{s}=\frac{n_0}{(\sqrt{2\pi}v_{th,s})^3}e^{-\beta_sH_s}\times\nonumber\\
&&\sum_{m=0}^\infty a_{2m}\left(\frac{\delta_s}{\sqrt{2}}\right)^{2m}H_{2m}\left(\frac{p_{xs}}{\sqrt{2}m_{s}v_{th,s}}\right)\times\nonumber\\
&&\sum_{n=0}^\infty b_n{\rm sgn}(q_{s})^n\left(\frac{\delta_s}{\sqrt{2}}\right)^nH_n\left(\frac{p_{ys}}{\sqrt{2}m_{s}v_{th,s}}\right),\label{eq:gammaANWT}
\end{eqnarray}
for
\begin{eqnarray}a_{2m}&=&\frac{(-1)^m2^{2m}}{(2m)!}Y_{2m}\left(0\,\,\frac{-1}{2\beta_{pl}}\,,\,0\,,\, \frac{-1}{2\beta_{pl}}\, ,\, ...\,,\,0\, ,\, \frac{-1}{2\beta_{pl}}\right)\nonumber ,\\
b_n&=&\frac{2^m}{n!}Y_{n}\left(\frac{1}{\beta_{pl}}\,,\,\frac{1}{\beta_{pl}}\,,\, ...\,,\,\frac{1}{\beta_{pl}}          \right).
\end{eqnarray}
One can readily calculate the number density for this DF using standard integral results (\citet{Gradshteyn}) to be 
\[N_s(A_{x},A_{y})=n_0\sum_{m=0}^\infty a_{2m}\left(\frac{A_x}{B_0L}\right)^{2m}\sum_{n=0}^\infty b_{n}\left(\frac{A_y}{B_0L}\right)^n= P_0\frac{\beta_e\beta_i}{\beta_e+\beta_i}.\]

\subsection{Plots of the exponential `re-gauged' distribution function for the FFHS}
We now present plots for the DF given in equation (\ref{eq:gammaANWT}), for $\beta_{pl}=0.05$ and $\delta_e=\delta_i=0.03$. This value for $\beta_{pl}$ is substantially lower than the value used in \citet{Allanson-2015PoP}, which had $\beta_{pl}=0.85$. The ability to go down to lower values of the plasma beta is due to the re-gauging process as explained in Subsection \ref{Subsec:regauge}. The plots that we show are intended to demonstrate progress in the numerical evaluation of low-beta DFs for nonlinear force-free fields, and as a proof of principle. Note that whilst the re-gauging process has allowed us to attain numerical convergence for low values of $\beta_{pl}$, the DF is proven to be convergent for all values of the relevant parameters.

The value of $\delta_s$ is chosen such that $\delta_s<\beta_{pl}$, since as explained in \citet{Allanson-2015PoP}, attaining convergence numerically has not been easy for values of $\delta_s>\beta_{pl}$ when $\beta_{pl}<1$.

Initial investigations of the shape of the variation of the DF in the $v_x$ and $v_y$ directions indicate that the DF seems to have a Gaussian profile, as in the DFs analysed in \citet{Allanson-2015PoP}. Hence, as in that work, we shall compare the DFs calculated in this work to `flow-shifted' Maxwellians, 
\begin{equation}
f_{Maxw,s}=\frac{n_0}{(\sqrt{2\pi}v_{th,s})^3}\exp\left[\frac{\left(\mathbf{v}-\langle\mathbf{v}\rangle_s(z)\right)^2}{2v_{th,s}^2}\right], \label{eq:Maxshift}
\end{equation}   
in order to measure the actual difference between the Vlasov equilibrium $f_{s}$, and the Maxwellian $f_{Maxw,s}$. The above distribution reproduces identical zeroth- and first-order moments (as functions of $z$) as the DF defined by equation (\ref{eq:gammaANWT}), namely $n_0$ and $n_0\langle\mathbf{v}\rangle_s$. However, unlike the DF derived in this paper, $f_{Maxw,s}$ is not a solution of the Vlasov equation and hence not an equilibrium solution. For examples of using `flow-shifted' Maxwellians in kinetic simulations, see \citet{Hesse-2005,Guo-2014}.

In figures (\ref{fig:1a}-\ref{fig:1e}) and (\ref{fig:2a}-\ref{fig:2e}) we give contour plots in $(v_x/v_{th,s},v_y/v_{th,s})$ space of the `raw' difference between the DFs defined by equation (\ref{eq:gammaANWT}) and (\ref{eq:Maxshift}). These figures bear close resemblance to those presented in \citet{Allanson-2015PoP}. Specifically, we see `shallower' peaks for the exact Vlasov solution, $f_{s}$, than for $f_{Maxw,s}$. There is also a clear anisotropic effect in that $f_{s}$ falls of more quickly in the $v_x$ direction than the $v_y$ direction as compared to $f_{Maxw,s}$. Note that whilst the raw differences plotted in these figures may not seem substantial, they can in fact be substantial as a proportion of $f_{Maxw,s}$, and even of the order of the magnitude of $f_{Maxw,s}$. As a demonstration of this fact we present plots in figures (\ref{fig:3a}-\ref{fig:3e}) and (\ref{fig:4a}-\ref{fig:4e}) of the quantity defined by
\[
f_{diff,s}=(f_{s}-f_{Maxw,s})/f_{Maxw,s}
\]
for line cuts through $(v_x/v_{th,s},v_y/v_{th,s}=0)$ and $(v_x/v_{th,s}=0,v_y/v_{th,s})$ respectively, for the ions. As suggested by the contour plots, $f_{diff,i}$ takes on significantly larger values in the $v_y$ direction, indicating that the tail of $f_i$ falls off less quickly than $f_{Maxw,i}$ in $v_y$ than in $v_x$.

We are yet to observe multiple peaks in the multiplicative DFs for the FFHS, derived herein and in \citet{Allanson-2015PoP}. However, the summative Harrison-Neukirch equilibria (\citet{Harrison-2009b}) could develop multiple maxima for sufficiently large values of the magnitude of the drift velocities. For the DF derived in this paper, and as in \citet{Allanson-2015PoP}, the `amplitude' of the drift velocity profile across the current sheet is given by
\begin{equation}
\frac{u_{s}}{v_{th,s}}=2\text{sgn}(q_{s})\frac{\delta_s }{\beta_{pl}},\nonumber
\end{equation}
where $u_s$ represents the maximum value of the drift velocities. As a result, large values of the drift velocity correspond to large values of $\delta_s/\beta_{pl}$, and these are exactly the regimes for which we are struggling to attain numerical convergence. This theory suggests that we may not be seeing DFs with multiple maxima because we are not in the appropriate parameter space.

 \section{Illustrative case for a non-force-free magnetic field}\label{Sec:Channell}
The work in this paper was initially motivated by attempts to find DFs for force-free equilibria  ($\boldsymbol{j}\times\boldsymbol{B}=\nabla \mathsfbi{P}_{zz}=\boldsymbol{0}$). However there is nothing in the formal solution method for the inverse problem $\mathsfbi{P}_{zz}(A_{x},A_{y}) \to g_{s}(p_{xs},p_{ys})$ that requires the magnetic field under consideration to be force-free. Here we give an example of the use of the solution method to a pressure function that was first discussed in \citet{Channell-1976}. In that paper, Channell actually solved the inverse problem by the Fourier transform method, and showed that the solution was valid given certain restrictions on the parameters. We tackle the problem via the Hermite Polynomial method, and find that for the resultant DF to be convergent, we require exactly the same restrictions as Channell. This parity between the validity of the two methods is reassuring, and implies that the necessary restrictions on the parameters are in a sense `method independent', and are the result of fundamental restrictions on the inversion of Weierstrass transformations. 

The magnetic field considered by Channell is of the form \[\boldsymbol{B}=(B_{x}(z)\,,0\,,0),\] with a pressure function \[\mathsfbi{P}_{zz}=P_0e^{-\gamma \tilde{A}_{y}^{2}}\] for $\tilde{A}_{y}=A_{y}/(B_{0}L)$ and $\gamma>0$ dimensionless. Note that the $\gamma$ used by Channell has dimensions equivalent to $1/(B_{0}^{2}L^{2})$. Note also that since the pressure is not constant, $P_0$ does not represent the value of the pressure, rather it is just some reference value. We can now write the details of the inversion. The equation we must solve, for a DF given by \[f_{s}=\frac{n_{0}}{(\sqrt{2\pi}v_{th,s})^{3}}e^{-\beta_{s}H_{s}}g_{s}(p_{ys};v_{th,s})\]  is
\[
P_0\exp \left(-\gamma \frac{A_{y}^{2}}{B_{0}^{2}L^{2}}\right)=\frac{n_0 (\beta_e+\beta_i)}{\beta_e \beta_i} \frac{1}{\sqrt{2\pi}m_{s}v_{th,s}}\int_{-\infty}^\infty {\rm e}^{-(p_{ys}-q_{s}A_y)^2/(2m_{s}^2v_{th,s}^2)}g_{s}dp_{ys}.
\]
We can immediately formally invert this equation as per the methods described in this paper, given the Macluarin expansion of the pressure 
\[ \mathsfbi{P}_{zz}=P_0\sum_{m=0}^\infty a_{2m}\left(\frac{A_y}{B_0L}\right)^{2m}\hspace{3mm}\text{s. t.} \hspace{3mm}a_{2m}=\frac{(-1)^m\gamma ^m}{m!},\]
to give
\[g_{s}(p_{ys})=\sum_{m=0}^\infty \left(\frac{\delta_s}{\sqrt{2}}\right)^{2m}a_{2m}H_{2m}\left(\frac{p_{ys}}{\sqrt{2}m_{s}v_{th,s}}\right).\]
Let us turn to the question of convergence. Theorem 1 states that if
\[\lim_{m\to\infty} m\left|\frac{a_{2m+2}}{a_{2m}}\right|<1/(2\delta_s^2), \]
then the $g_{s}$ function is convergent. This is readily seen to imply that if $\gamma$ satisfies
\[\gamma< \frac{1}{2\delta_s^2},\]
then the Hermite series representation for  $g_{s}$ is convergent. This condition is exactly equivalent to the one derived by Channell (equation (28) in the paper). Note that now that we have established convergence for particular $\gamma$, then boundedness results follow as per other results given in this paper, detailed in Appendix A. One more question remains, namely how does the $g_{s}$ function derived compare to the Gaussian $g_{s}(p_{ys})$ function derived by Channell 
\[g_{s}\propto e^{-4\gamma^2\delta_s^4p_{ys}^2/(1-4\gamma^2\delta_s^4)}\]
(in our notation) using the method of Fourier transforms? In fact, one can see by setting $y=0$ in Mehler's Hermite Polynomial formula (\citet{Watson-1933}) 
\[\frac{1}{\sqrt{1-\rho^2}}\exp\left[\frac{2xy\rho-(x^2+y^2)\rho^2    }{1-\rho^2}   \right]=\sum_{n=0}^{\infty}\frac{\rho ^n}{2^n n!} H_n(x)H_n(y)  ,  \]
and using 
\begin{equation}
H_m(0)=
\begin{cases}
0 & \text{if}\hspace{3mm} m\hspace{3mm} \text{is odd}, \\
  (-1)^{m/2}m!/(m/2)!      & \text{if}\hspace{3mm} m\hspace{3mm} \text{is even},
\end{cases}\nonumber
\end{equation}
(see \citet{Gradshteyn} for example) we see that the Hermite series represents a Gaussian function in the range $| \rho |<1$. This is equivalent to the condition derived above for convergence, $\gamma< 1/(2\delta_s^2)$. Hence, we have shown that for this specific example - solvable by using both Hermite polynomials and Fourier transforms - the two methods used to solve the inverse problem give equivalent functions with equivalent ranges of validity.

\section{Summary}
The primary result of this paper is the rigorous generalisation of a solution method that exactly solves the `inverse problem' in 1-D collisionless equilibria, for a certain class of equilibria. Specifically, given a pressure function, $\mathsfbi{P}_{zz}(A_x,A_y)$, of a separable form, neutral equilibrium distribution functions can be calculated that reproduce the prescribed macroscopic equilibrium, provided $\mathsfbi{P}_{zz}$ satisfies certain conditions on the coefficients of its (convergent) Maclaurin expansion, and is itself positive. In particular, for force-free magnetic fields, there is an algorithmic path taking the magnetic field, $\boldsymbol{B}(\boldsymbol{A})$, as input, and giving the distribution function $f_{s}$ as output. 

The distribution function has the form of a Maxwellian modified by a function $g_s$, itself represented by -- possibly infinite -- series of Hermite polynomials in the canonical momenta. It is crucial that these series are convergent and positive for the solution to be meaningful. A sufficient condition was derived for convergence of the distribution function by elementary means, namely the ratio test, with the result a restriction on the rate of decay of the Maclaurin coefficients of $\mathsfbi{P}_{zz}$. We also argue that for such a pressure function that is also positive, that the Hermite series representation of the modification to the Maxwellian is positive, for sufficiently low values of the magnetisation parameter, i.e. lower than some critical value. This was actually proven for a certain class of $g_s$ functions, and differentiability of $g_s$ was assumed. It would be interesting in the future to investigate whether this critical value of the magnetisation parameter can be determined. Note that whilst we have not yet determined the critical value, we have not yet observed negative distribution functions for the pressure functions and parameter ranges studied.

Examples of the use of the Hermite Polynomial method are given for DFs that correspond to the force-free Harris Sheet, including calculations for a DF with a different gauge to that considered previously, motivated by numerical reasons. We have presented some plots of a comparison between the re-gauged DFs and shifted Maxwellian functions, as a proof of principle, namely that numerical convergence for values of $\beta_{pl}$ lower than previously reached, can now be attained ($\beta_{pl}=0.05$). Verification of the analytical properties of convergence and boundedness of the distribution functions written as infinite sums over Hermite polynomials are given in appendix A. Note that the verification of these distribution functions is rather involved due to the complex nature of the specific Maclaurin expansions that we consider, and is simpler for more `straightforward' expansions, e.g. for the example considered in Section \ref{Sec:Channell}. 

We have demonstrated the application of the solution method presented in this paper to the force-free Harris Sheet magnetic field. However, potential uses go beyond this example, including magnetic fields that are not force-free. To this end we consider a non-force-free example in Section \ref{Sec:Channell}. This particular example already has a known solution and range of validity in parameter space, obtained by a Fourier transform method in \citet{Channell-1976}. We obtain a solution with an alternate representation using the Hermite Polynomial method. The Hermite series obtained is shown to be equivalent to the representation obtained by Channell, and to have the exact same range of validity in parameter space. It is not clear if this equivalence between solutions obtained by the two different methods is true in general. Our problem is somewhat analagous to the heat/diffusion equation, and in that `language' the question of the equivalence of solutions is related to the `backwards uniqueness of the heat equation' (see e.g. \citet{Evansbook}). The degree of similarity between our problem and the one described by Evans, and its implications, are left for future investigations.

Also, whilst we have assumed that the pressure is separable (either summatively or multiplicatively), the method should be adaptable in the `obvious way' for pressures that are a combination of the two types. Interesting further work would be to see if the method can be adapted to work for pressure functions that are non-separable, i.e. of the form
\begin{equation*}
\mathsfbi{P}_{zz}=\sum_{m,n}\mathcal{C}_{mn}\left(\frac{A_x}{B_0L}\right)^{m}\left(\frac{A_y}{B_0L}\right)^{n}.
\end{equation*} 
This would be pertinent for pressure tensors transformed in such a way that they are no longer separable.

Other future work could involve an in-depth parameter study of the new re-gauged multiplicative distribution function for the FFHS, with an analysis of how far the exact equilibrium distribution function differs from an appropriately flow-shifted Maxwellian, frequently used in fully kinetic simulations for reconnection studies. In particular it would be interesting to see how much the distribution functions differ from flow-shifted Maxwellians as the set of parameters $(\beta_{pl}, \delta_{s})$ are varied across a wide range. Preliminary numerical investigations verify that plotting distribution functions for the FFHS with a lower $\beta_{pl}$ than previously achieved, namely $\beta_{pl}=0.05$ rather than $\beta_{pl}=0.85$, has been made possible by the theoretical developments in this paper. We have not yet observed multiple maxima for the distribution functions, but do see significant deviations from Maxwellian distributions, and an anisotropy in velocity space.

\section*{}
The authors would like to thank the anonymous referees, whose comments have significantly improved the manuscript. OA would also like to acknowledge helpful correspondence with Professor Wilfrid S Kendall, University of Warwick.

The authors gratefully acknowledge the financial support of the Leverhulme Trust [F/00268/BB] (TN \& FW), a Science and Technology Facilities Council Consolidated Grant [ST/K000950/1] (TN \& FW), a Science and Technology Facilities Council Doctoral Training Grant [ST/K502327/1] (OA) and an Engineering and Physical Sciences Research Council Doctoral Training Grant [EP/K503162/1] (ST). The research leading to these results has received funding from the European
Commission's Seventh Framework Programme FP7 under the grant agreement SHOCK [284515] (OA, TN \& FW). 

\appendix
\section{Convergence and boundedness of the FFHS DFs}\label{app:A}
\subsection{Multiplicative DF for the FFHS in the `original' gauge: $\beta_{pl}\in (0\,,\,\infty)$ }
\subsubsection{Convergence of the Hermite representation of $g_s$}
Here we include the full details of the calculations that confirm the validity of the Hermite Polynomial representation of the multiplicative FFHS equilibrium in original gauge (\citet{Allanson-2015PoP}), for the first time. We shall first verify the convergence of $g_{2s}$ (expanded over $n$ in equation (\ref{eq:result})) using the convergence condition from Subsection \ref{sec:convergence}, and then verify convergence of $g_{1s}$ by comparison with $g_{2s}$. As Theorem 1 states, we can verify convergence of $g_{2s}$ provided
\[
\lim_{{n\to\infty}}n\left|\frac{b_{n+1}}{b_n}\right|<1/\delta_s
.\]
Explicit expansion of the exponentiated exponential series by `twice' using Maclaurin series (as opposed to the CBP formulation of equation (\ref{eq:Bell})) gives
\begin{equation}
b_n=\frac{2^n}{n!}\sum_{k=0}^\infty \frac{k^n}{\beta_{pl}^kk!}\label{eq:bmac}
\end{equation}
and so
\begin{eqnarray*}
b_{n+1}/b_n&=&\frac{2}{n+1}\sum_{j=1}^\infty \frac{j^n}{(j-1)!\beta_{pl}^{j}}\biggr/\sum_{j=1}^\infty\frac{j^n}{j!\beta_{pl}^{ j}}\\
&=&\frac{2}{n+1}\left(\frac{\displaystyle \frac{1}{0!\beta_{pl}}+\frac{2^n}{1!\beta_{pl}^{ 2}}+\frac{3^n}{2!\beta_{pl}^{ 3}}+...}{\displaystyle \frac{1}{1!\beta_{pl}}+\frac{2^n}{2!\beta_{pl}^{ 2}}+\frac{3^n}{3!\beta_{pl}^{ 3}}+...}\right)\\
&=&\frac{2}{n+1}\left(\frac{\displaystyle \frac{1}{\beta_{pl}}+2\frac{2^n}{2!\beta_{pl}^{ 2}}+3\frac{3^n}{3!\beta_{pl}^{ 3}}+...}{\displaystyle \displaystyle \frac{1}{1!\beta_{pl}}+\frac{2^n}{2!\beta_{pl}^{ 2}}+\frac{3^n}{3!\beta_{pl}^{ 3}}+...}\right).
\end{eqnarray*}
The $k$th `partial sum' of this fraction has the form
\[
r_k=\frac{p_1+2p_2+3p_3+...+kp_k}{p_1+p_2+p_3+...}
\]
with $p_i\asymp 1/i!$, where we write $g\asymp h$ to mean $g/h$ and $h/g$ are bounded away from $0$.
Now since the denominator of the $p_{i}$ increase super-exponentially (factorially) we have $ i p_{i}\asymp p_{i}$ and hence
\[
0<\sum_{i=1}^{\infty}ip_{i}<\infty \hspace{3mm}{\rm and }\hspace{3mm} 0<\sum_{i=1}^{\infty}p_{i}<\infty.
\]
Thus $r_{k}\to r_{\infty}\in (0,\infty)$ and, more specifically, $r_{\infty}\asymp 1$ in $n$. 
Therefore 
\[
b_{n+1}/b_{n}= r_{\infty}/(n+1)\asymp 1/n.
\]
That is to say $b_{n+1}/b_{n}$ behaves asymptotically like $1/n$. This satisfies the condition of Theorem 1. Hence $g_{2s}(p_{ys})$ converges  for all $\delta_s$ and $p_{ys}$ by the comparison test.  

We shall now verify convergence of $g_{1s}$, by comparison with $g_{2s}$. By explicitly using the Maclaurin expansion of the exponential, and then the power-series representation for $\cos^nx$ from \citet{Gradshteyn}
\begin{eqnarray*}
\cos ^{2n}x&=&\frac{1}{2^{2n}}\left[\sum_{k=0}^{n-1}2{2n \choose k}\cos (2(n-k)x)+{2n \choose n} \right],\\
\cos ^{2n-1}x&=&\frac{1}{2^{2n-2}}\sum_{k=0}^{n-1}{2n-1 \choose k}\cos ((2n-2k-1)x),
\end{eqnarray*}
one can calculate 
\[
\exp\left(\frac{1}{2\beta_{pl}}\cos\left(\frac{2A_x}{B_0L}\right)\right)=\sum_{m=0}^{\infty}a_{2m}\left(\frac{A_x}{B_0L}\right)^{2m}.
\]
The zeroth coefficient is given by $a_0=\exp\left(1/(2\beta_{pl})\right)$, and the rest are
\[
a_{2m}=\frac{2(-1)^m}{(2m)!}\sum_{k=0}^\infty\sum_{j\in J_{k}} \frac{1}{j!(4\beta_{pl})^j}{j \choose k}(j-2k)^{2m}, \label{eq:appendix1}
\]
for $J_{k}=\left\{2k+1, 2k+2, ...\right\}$ and $m\neq 0$. By rearranging the order of summation, $a_{2m}$ can be written 
\[
a_{2m}=\frac{2(-1)^m}{(2m)!}\sum_{j=1}^\infty \frac{1}{j!(4\beta_{pl})^j}\sum_{k=0}^{\lfloor (j-1)/2 \rfloor}{j \choose k}(j-2k)^{2m}, 
\]
where $\lfloor x \rfloor$ is the floor function, denoting the greatest integer less than or equal to $x$. Recognising an upper bound in the expression for $a_{2m}$;
\[
\sum_{n=0}^{\lfloor (j-1)/2 \rfloor}{j \choose n}(j-2n)^{2m}\leq j^{2m}\sum_{n=0}^j{j \choose n}=2^{j}j^{2m},
\]
gives 
\begin{eqnarray*}
a_{2m}<\frac{2(-1)^m}{(2m)!}\sum_{j=1}^\infty\frac{2^{j+1}j^{2m}}{j!2^j(2\beta_{pl})^j}&=&2\frac{(-1)^m}{(2m)!}\sum_{j=1}^\infty\frac{j^{2m}}{j!(2\beta_{pl})^j},\\
&\le & \frac{2}{(2m)!}\sum_{j=1}^\infty\frac{j^{2m}}{j!(2\beta_{pl})^j},\\
&=&\frac{1}{(2m)!}\sum_{j=1}^\infty\frac{2^{1-j}j^{2m}}{j!\beta_{pl}^j}<b_{2m}
\end{eqnarray*}
Hence we now have an upper bound on $a_{2m}$ for $m\neq 0$ and we know that $a_{2m+1}=0$, and so is bounded above by $b_{2m+1}$. Note also that $a_0<b_0$. Hence, each term in our series for $g_{1s}(p_{xs})$ is bounded above by a series known to converge for all $\delta_s$ according to
\[
a_l\left(\frac{\delta_s}{\sqrt{2}}\right)^lH_l(x)<b_l\left(\frac{\delta_s}{\sqrt{2}}\right)^lH_l(x).
\]
So by the comparison test, we can now say that $g_{1s}\left(p_{xs}\right)$ is a convergent series. Hence the representation of the DF in equation (\ref{eq:result}) is convergent. 

%\subsubsection{Non-negativity}
%Theorem 2. states that the DF is non-negative, provided the Maclaurin coefficients of the pressure, $a_j$, satisfy
%\[
%\sum_{j=0}^\infty \left(\frac{\delta_s}{\sqrt{2}}\right)^{2j}(-1)^j\frac{(2j)!}{j!}a_{2j}>0.
%\]
%First we can note that the $a_{2m}$ coefficients for the exponentiated cosine series have a sign given by $(-1)^m$. Hence they trivially satisfy the above inequality as every term in the above series has positive sign, and we can conclude that the $g_{1s}$ series is non-negative by Theorem 2. The exponentiated exponential series requires more careful thought. We see from (\ref{eq:hermite0}), (\ref{eq:result}) and (\ref{eq:bmac}) that 
%\begin{eqnarray}
%g_{2s}(0)&=&\sum_{k=0}^\infty \frac{1}{\beta_{pl}^k k!}\sum_{n=0}^\infty \left(2\delta_s^2\right)^{n}(-1)^n\frac{1}{n!}k^{2n}\\
%&=&\sum_{k=0}^\infty \frac{1}{k!\beta_{pl}^k}e^{-2\delta_s^2k^2}>0
%\end{eqnarray}
%Hence we can confirm that $g_{2s}$ is non-negative, and so is the entire DF.

\subsubsection{Boundedness of the `original' gauge DF}
Since $g_{1s}$ and $g_{2s}$ are known to be convergent, we know that for a given $z$, the DF is bounded in momentum space by
\begin{eqnarray*}
|f_{s}|&<&\; {\rm e}^{-\beta_sH_s}\exp\left(\frac{p_{xs}^2}{4m_{s}^2v_{th,s}^2}+\frac{p_{ys}^2}{4m_{s}^2v_{th,ss}^2}\right)S_{1s}S_{2s},\\
&=&\; {\rm e}^{-\left(\frac{1}{2}(p_{xs}^2+p_{ys}^2)-2q_{s}(p_{xs}A_x+p_{ys}A_y)+q_{s}^2(A_x^2+A_y^2) \right)/(2m_{s}^2v_{th,s}^2)}S_{1s}S_{2s}
\end{eqnarray*} 
where $S_{1s}$ and $S_{2s}$ are finite constants. The `additional' exponential factors come from the upper bounds on Hermite polynomials used in equation (\ref{eq:hermbound}). This clearly goes to zero for sufficiently large $|p_{xs}|$, $|p_{ys}|$ and is without singularity. We conclude that the distribution is bounded/normalisable.

\subsection{ Multiplicative DF for the `re-gauged' FFHS: $\beta_{pl}\in (0\,,\,\infty)$  }
\subsubsection{Convergence of the Hermite representation of $g_s$}
This DF has the exact same coefficients for the $p_{ys}$-dependent Hermite polynomials as that discussed above. And so we need not verify convergence for that series. And in fact, all that has changed in the analysis of the coefficients for the $p_{xs}$-dependent sum is that we now have to consider the Maclaurin coefficients of $\sin ^2(A_x/(B_0L))$ as opposed to $\cos(2A_x/(B_0L))$. These Maclaurin coefficients both have the same factorial dependence and as such the convergence of the one DF implies the convergence of the other. 

%\subsubsection{Non-negativity}
%As mentioned above, the new DF has the exact same coefficients for the $p_{ys}$ dependent Hermite polynomials as that discussed above. And so we need not verify non-negativity for that series. The $g_{1s}$ series requires more thought however. We can write 
%\begin{eqnarray}
%\exp\left(\frac{1}{\beta_{pl}}\sin ^2\left(\frac{A_x}{B_0L}\right)\right)&=&\sum_{m=0}^\infty \frac{1}{\beta_{pl}^m}\frac{1}{m!}\sin^{2m}\left(\frac{A_x}{B_0L}\right),\nonumber\\ 
%\text{s.t.}\hspace{3mm}\sin^{2m}\left(\frac{A_x}{B_0L}\right)&=&\sum_{n=0}^\infty\sigma_{2n}(m)\left(\frac{A_x}{B_0L}\right)^{2n}.\nonumber
%\end{eqnarray}
%This then implies that 
%\begin{eqnarray}
%g_{1s}(0)&=&\sum_{m=0}^\infty\frac{1}{\beta_{pl}^m}\frac{1}{m!}\sum_{n=0}^\infty\left(\frac{\delta_s}{\sqrt{2\gamma}}\right)^{2n}\frac{(2n)!}{n!}(-1)^n\sigma_{2n}(m),\nonumber\\
%&=&\sum_{m=0}^\infty\frac{1}{\beta_{pl}^m}\frac{1}{m!}K_m.\nonumber
%\end{eqnarray}
%Clearly, if we can demonstrate that $K_m>0$ then we can prove non-negativity. An inspection of the Maclaurin series for $\sin ^{2m}(x)$ for $m=0,1,2,...$, shows that ${\rm sgn}(\sigma_{2n})=(-1)^n$, and as such $K_m>0$ trivially. Hence non-negativity of the DF is demonstrated.

\subsubsection{Boundedness of the `re-gauged' DF}
The boundedness argument is exactly analogous to that made above for the DF in original gauge, and need not be repeated here.

\section{On the lower bound of the $\bar{g}_s$ function}\label{App:nonneg}
Here we give some technical remarks that support our claim that $\bar{g}_s$ (and hence $g_s$) is bounded below, using an argument by contradiction.

First of all consider a smooth $\bar{g}_s$ function that is unbounded from below in positive momentum space. Then, depending on the number and nature of stationary points, either 
\begin{itemize}
\item Case 1: There will be some $\tilde{p}_0$ such that $\bar{g}_s<c<0$  for all $\tilde{p}_s>\tilde{p}_0$. This is a trivial statement if $\bar{g}_s$ has only a finite number of stationary points, whereas in the case of an infinite number of stationary points, all maxima of $\bar{g}_s$ for $\tilde{p}_s>\tilde{p}_0$ must be `away' from zero by a finite amount.
\item Case 2:  In this case the (infinite number of) maxima either can rise above zero, or tend to zero from below in a limiting fashion.
\end{itemize}

If $\bar{g}_s$ is of the type described in Case 1, then we can create an `envelope' $g_{env}$ for $\bar{g}_s$ such that $g_{env}>\bar{g}_s$ for all $\tilde{p}_s$. The envelope we choose is 
\begin{equation}
g_{env}=
\begin{cases}
Le^{\tilde{p}_s^2/4},\text{   for   } \tilde{p}_s\le \tilde{p}_0,\\
c\text{   for   } \tilde{p}_s > \tilde{p}_0.
\end{cases}
\end{equation}
We choose the $Le^{\tilde{p}_s^2/4}$ profile because this represents the absolute upper bound for our convergent Hermite expansions, at a given $\tilde{p}_s$ as seen from (\ref{eq:hermbound}). If we then substitute the $g_{env}$ function for $\bar{g}_s$ in (\ref{eq:newsample}) the integrals give combinations of error functions, from which it is seen that one obtains a negative result for sufficiently large $\tilde{A}$. This is a contradiction since the left-hand side of (\ref{eq:newsample}) is positive for all $\tilde{A}$. Hence we can discount the $\bar{g}_{s}$ functions of the variety described in Case 1, as we have a contradiction.

Case 2 is less simple to treat. The fact that there exists an infinite number of local minima and that the infimum of $\bar{g}_s$ is $-\infty$ implies that there exists an infinite sequence of points in momentum space, $\mathcal{S}_p=\{ \tilde{p}_k\, : \, k=1,2,3 ...\}$, that are local minima of $\bar{g}_s$, such that $\bar{g}_s(\tilde{p}_{k+1})<\bar{g}_s(\tilde{p}_{k})$. Essentially there are an infinite number of minima `lower than the previous one'. For sufficiently large $k=l$, we have that the magnitude of the minima is much greater than the width of the Gaussian, i.e.
\[
|\bar{g}_{s}(\tilde{p}_{l})| \gg 2\sqrt{2}.
\]
In this case the only way that the sampling of $\bar{g}_s$ described by (\ref{eq:newsample}) could give a positive result for a Gaussian centred on the minima is if $\bar{g}_s$ rapidly grew to become sufficiently positive, in order to compensate the negative contribution from the minimum and its local vicinity. However, this seems to be at odds with the condition that $\bar{g}_s$ is smooth, since the function would have to rise in this manner for ever more negative values of the minima (and hence rise ever more quickly) as $k\to \infty$. We claim that this can not happen, and hence we discount the $\bar{g}_{s}$ functions of the variety described in Case 2.

Since there is no asymmetry in momentum-space in this problem, the arguments above hold just as well for for a $\bar{g}_s$ function that is unbounded from below in negative momentum space. It should be clear to see that if $\bar{g}_s$ can not be unbounded from below in either the positive or negative direction, then it can not be unbounded in both directions either.

\bibliographystyle{jpp}

%\bibliography{/user/oliver/Dropbox/biblio}   %Maths
%\bibliography{/Users/Oliver/Dropbox/biblio}                  %Home windows
%\bibliography{/Users/OAllanson/Dropbox/biblio}  %Home mac
\newpage

\begin{figure}
\centering
\begin{subfigure}[b]{0.48\textwidth}
\includegraphics[width=\textwidth]{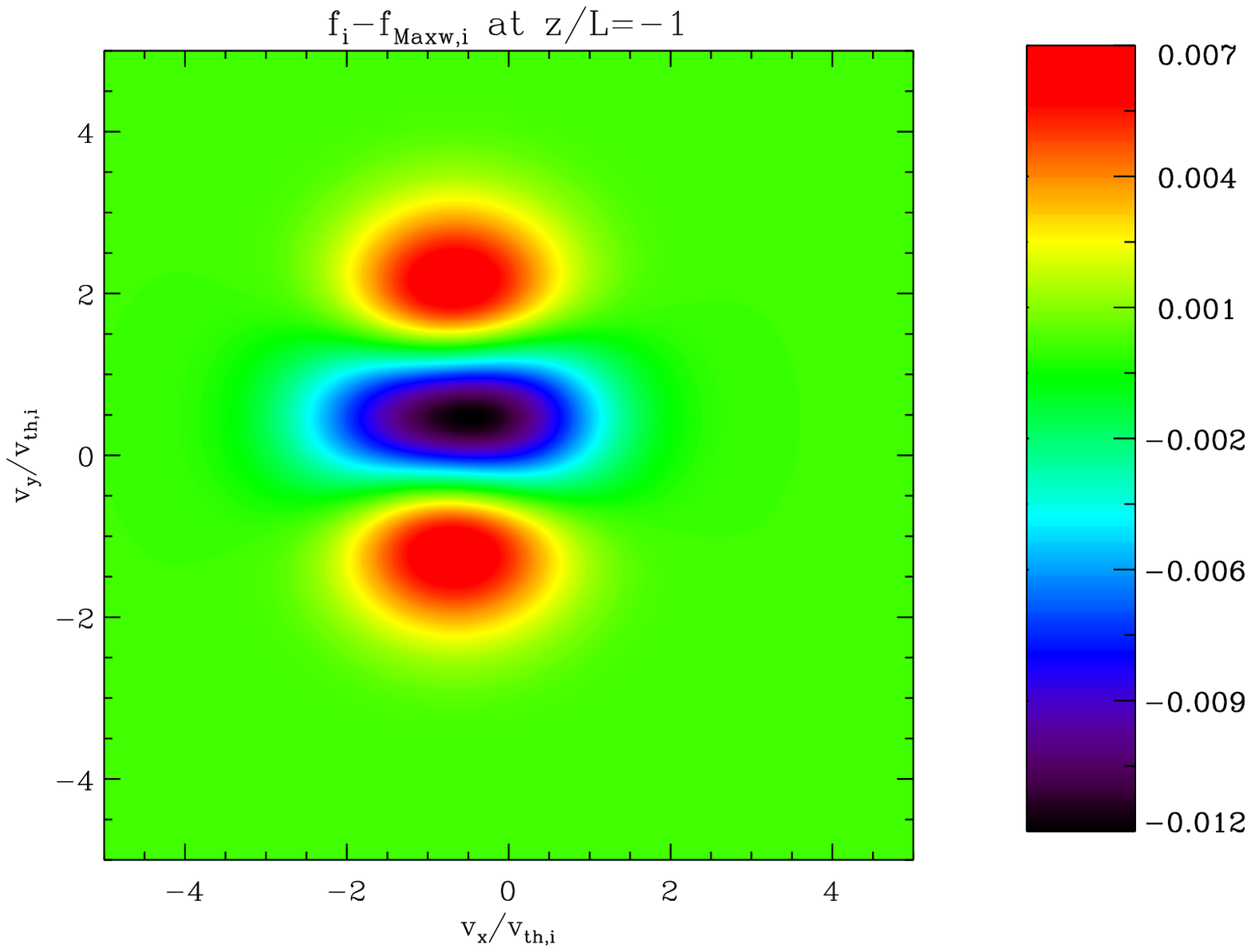}
 \caption{\label{fig:1a}}
\end{subfigure}
\begin{subfigure}[b]{0.48\textwidth}
\includegraphics[width=\textwidth]{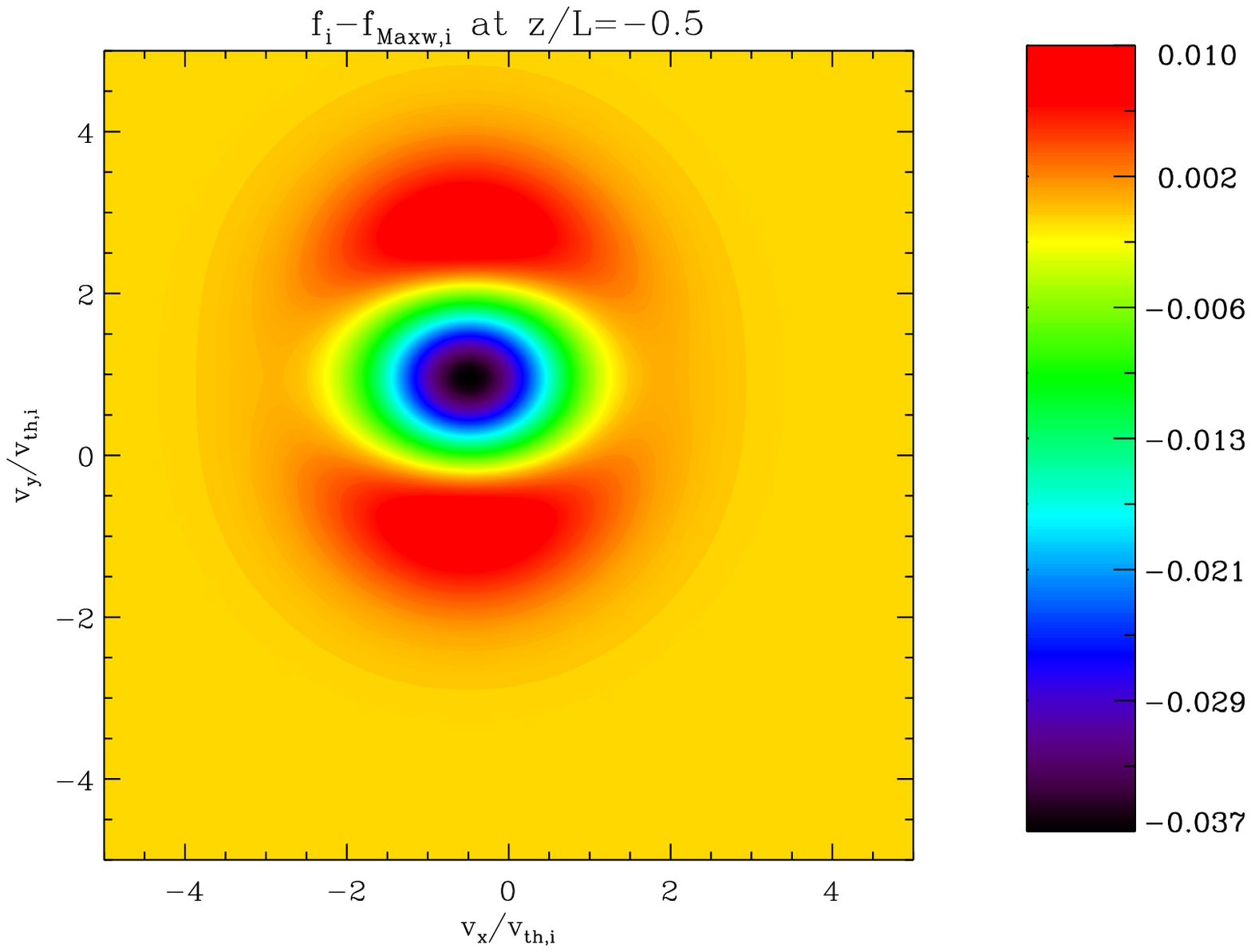}
 \caption{\label{fig:1b}}
\end{subfigure}
\begin{subfigure}[b]{0.48\textwidth}
\includegraphics[width=\textwidth]{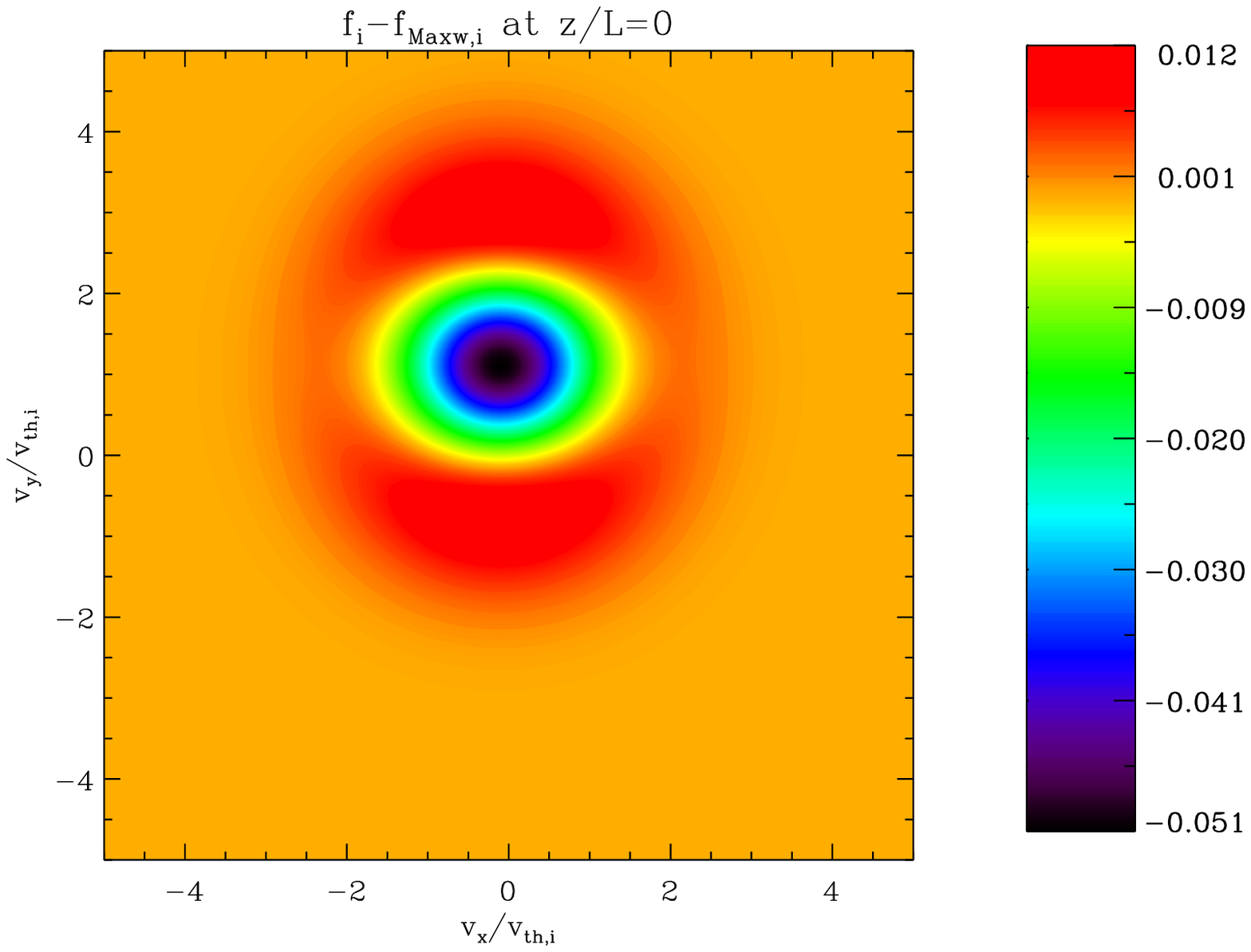}
 \caption{\label{fig:1c}}
\end{subfigure}
\begin{subfigure}[b]{0.48\textwidth}
\includegraphics[width=\textwidth]{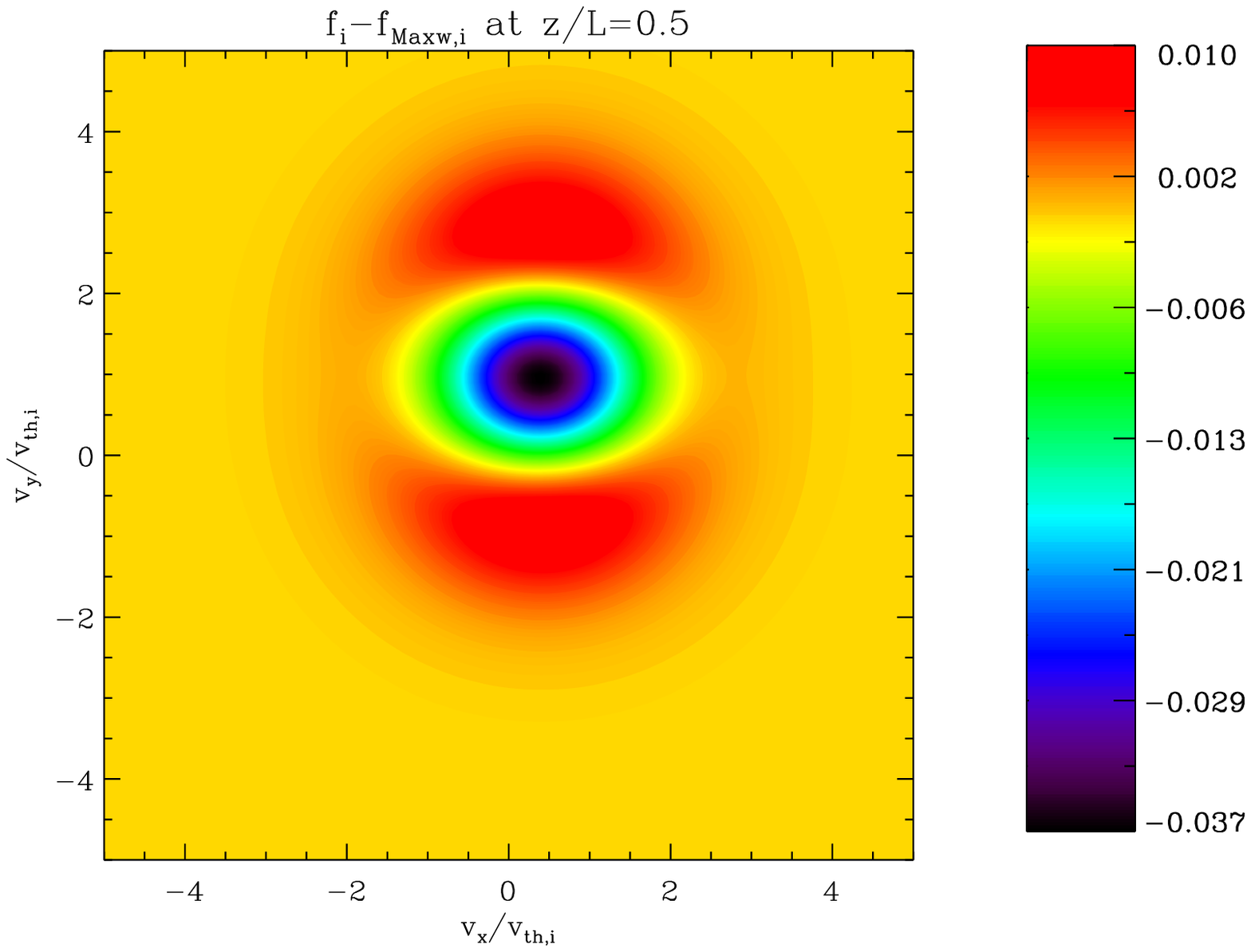}
 \caption{\label{fig:1d}}
\end{subfigure}
\begin{subfigure}[b]{0.48\textwidth}
\includegraphics[width=\textwidth]{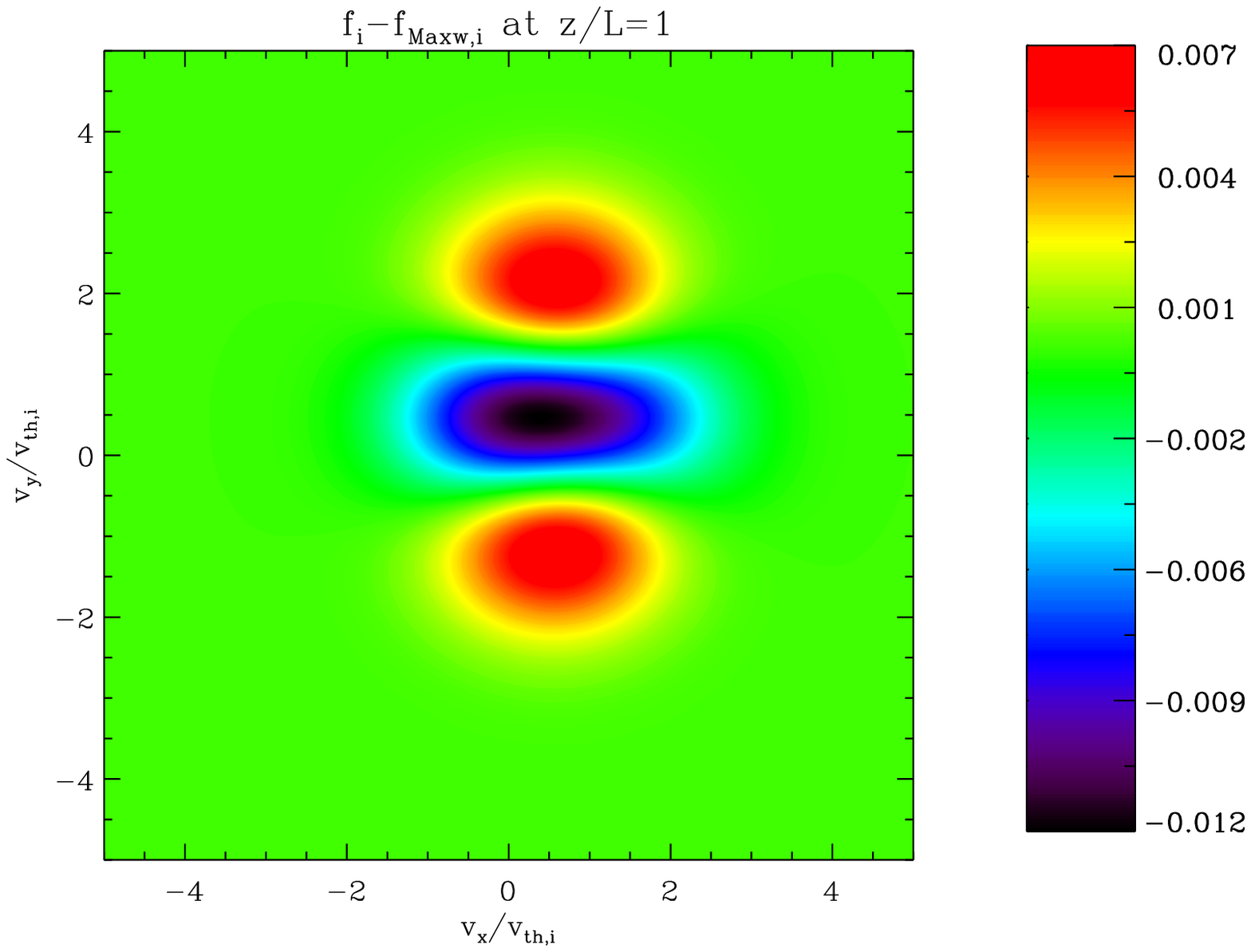}
 \caption{\label{fig:1e}}
\end{subfigure}
\caption{Contour plots of $f_i-f_{Maxw,i}$ for $z/L=-1$ (\ref{fig:1a}),  $z/L=-0.5$ (\ref{fig:1b}), $z/L=0$ (\ref{fig:1c}), $z/L=0.5$ (\ref{fig:1d}) and $z/L=1$ (\ref{fig:1e}). $\beta_{pl}=0.05$ and $\delta_i=0.03$. }
 \end{figure}

\begin{figure}
\centering
\begin{subfigure}[b]{0.48\textwidth}
\includegraphics[width=\textwidth]{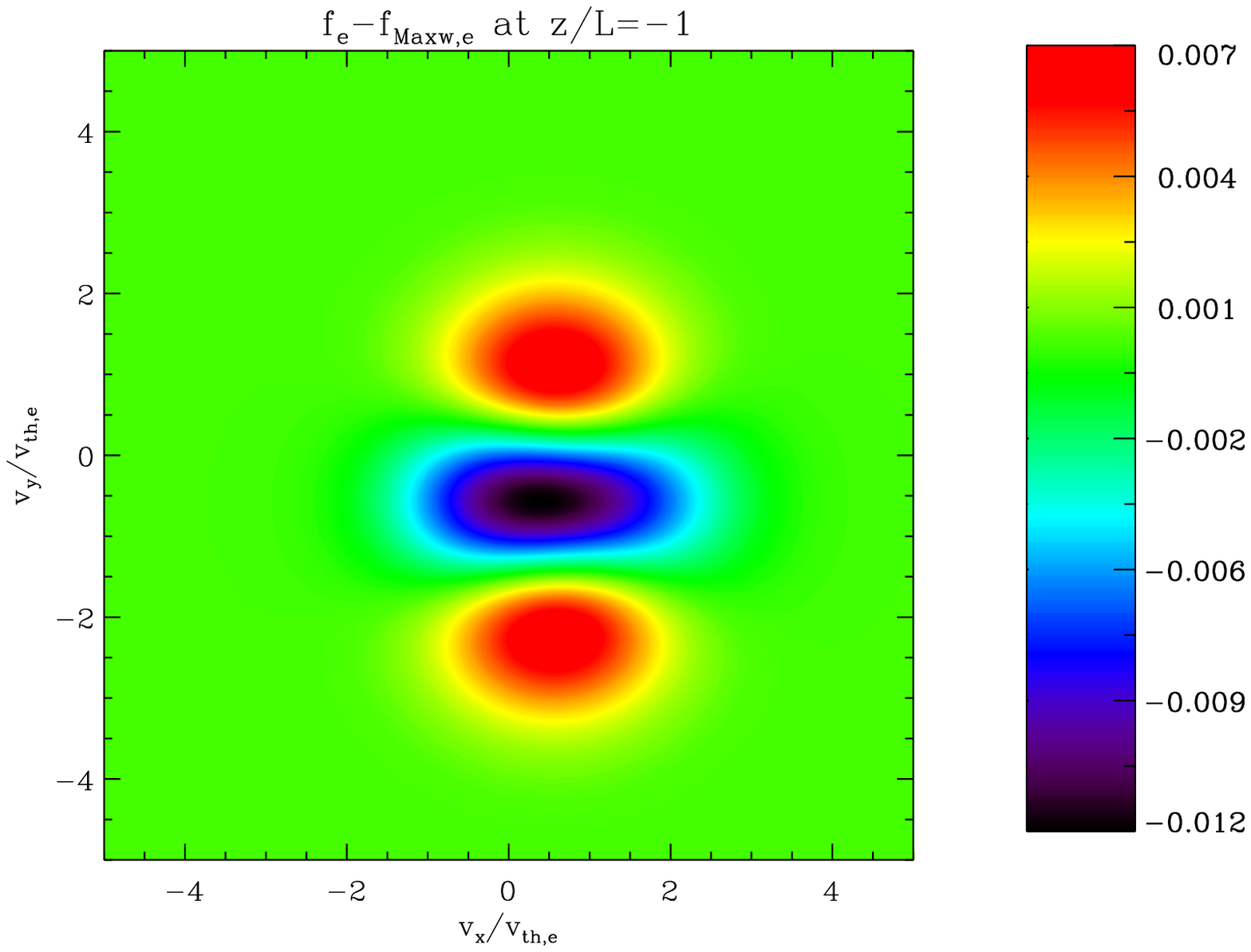}
 \caption{\label{fig:2a}}
\end{subfigure}
\begin{subfigure}[b]{0.48\textwidth}
\includegraphics[width=\textwidth]{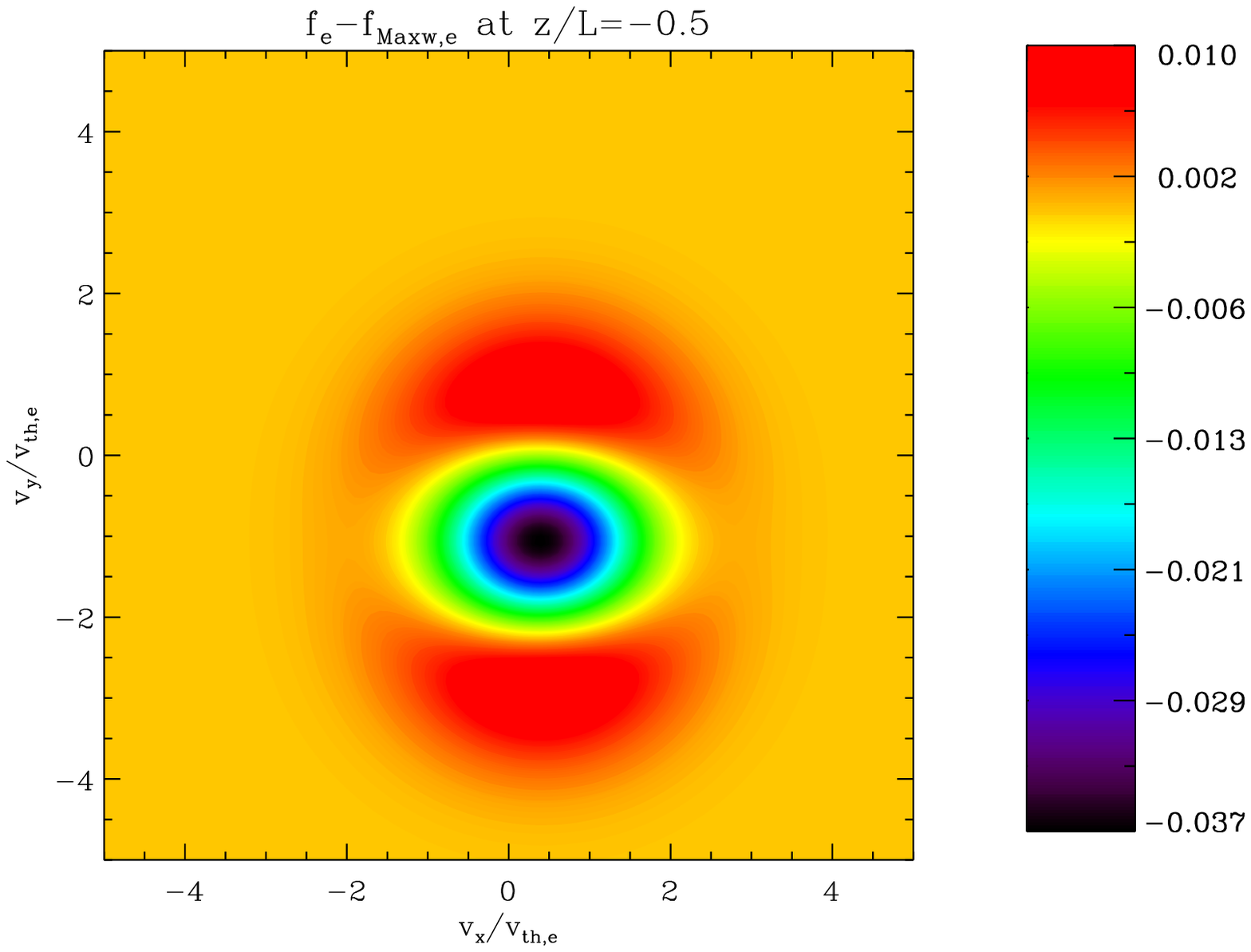}
 \caption{\label{fig:2b}}
\end{subfigure}
\begin{subfigure}[b]{0.48\textwidth}
\includegraphics[width=\textwidth]{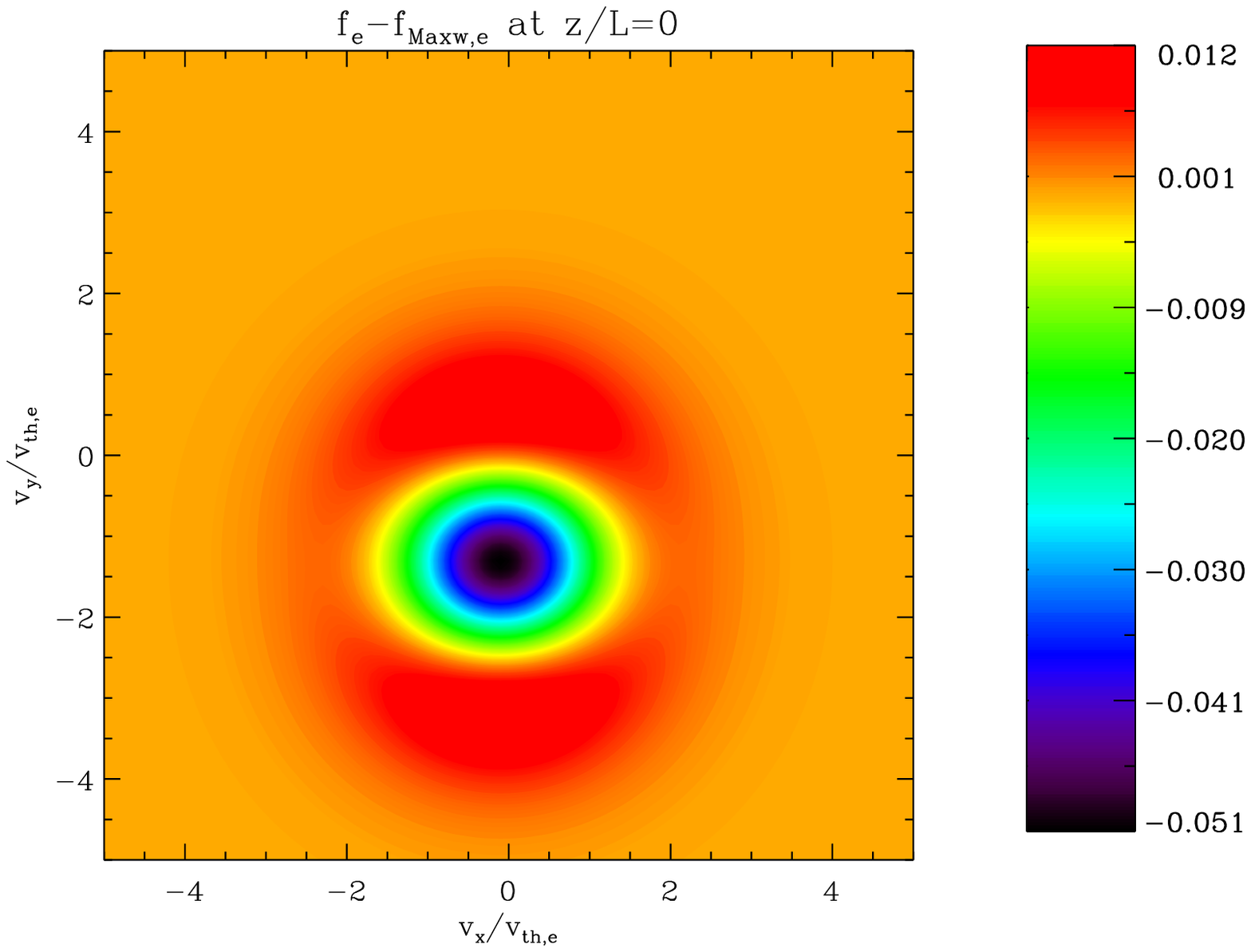}
 \caption{\label{fig:2c}}
\end{subfigure}
\begin{subfigure}[b]{0.48\textwidth}
\includegraphics[width=\textwidth]{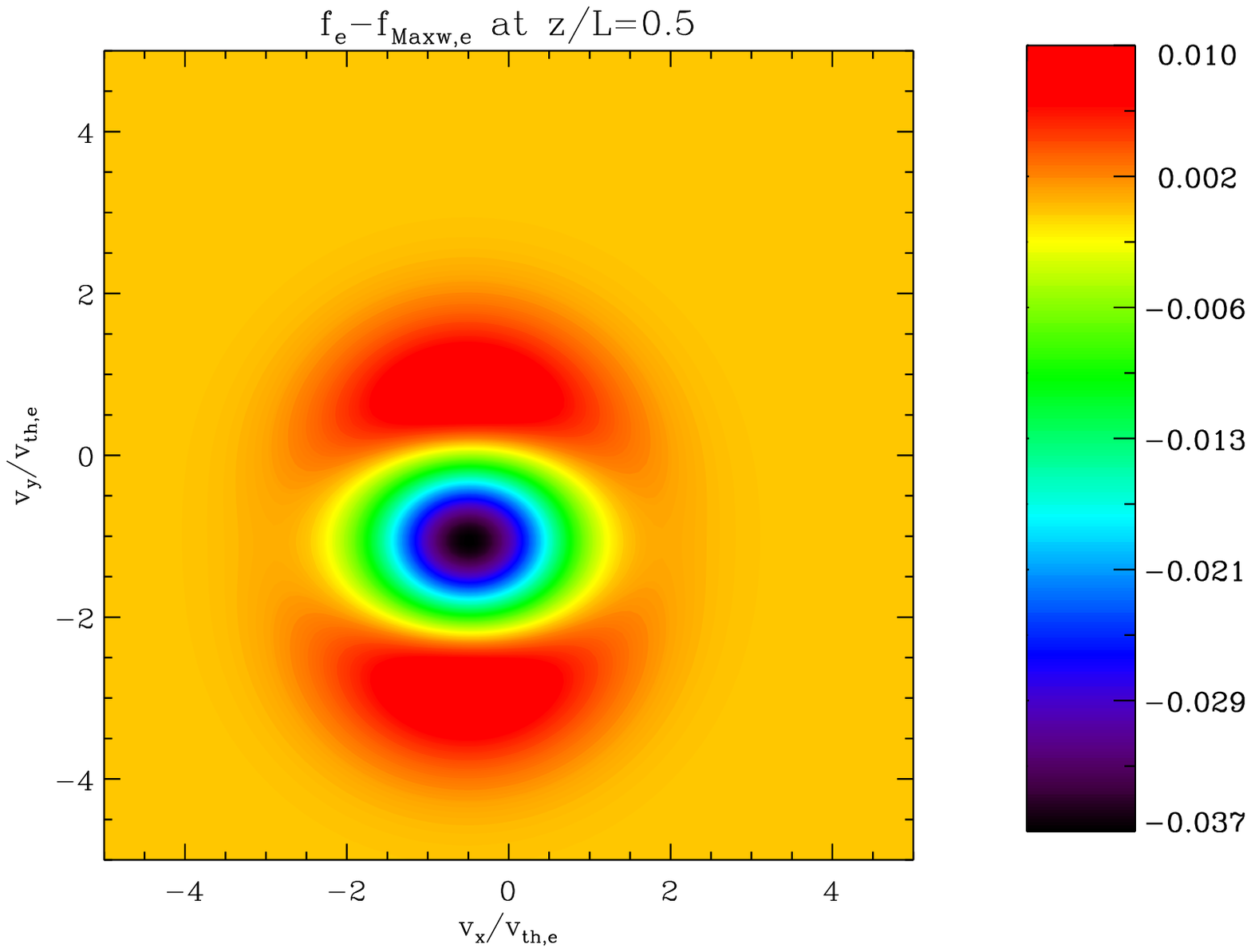}
 \caption{\label{fig:2d}}
\end{subfigure}
\begin{subfigure}[b]{0.48\textwidth}
\includegraphics[width=\textwidth]{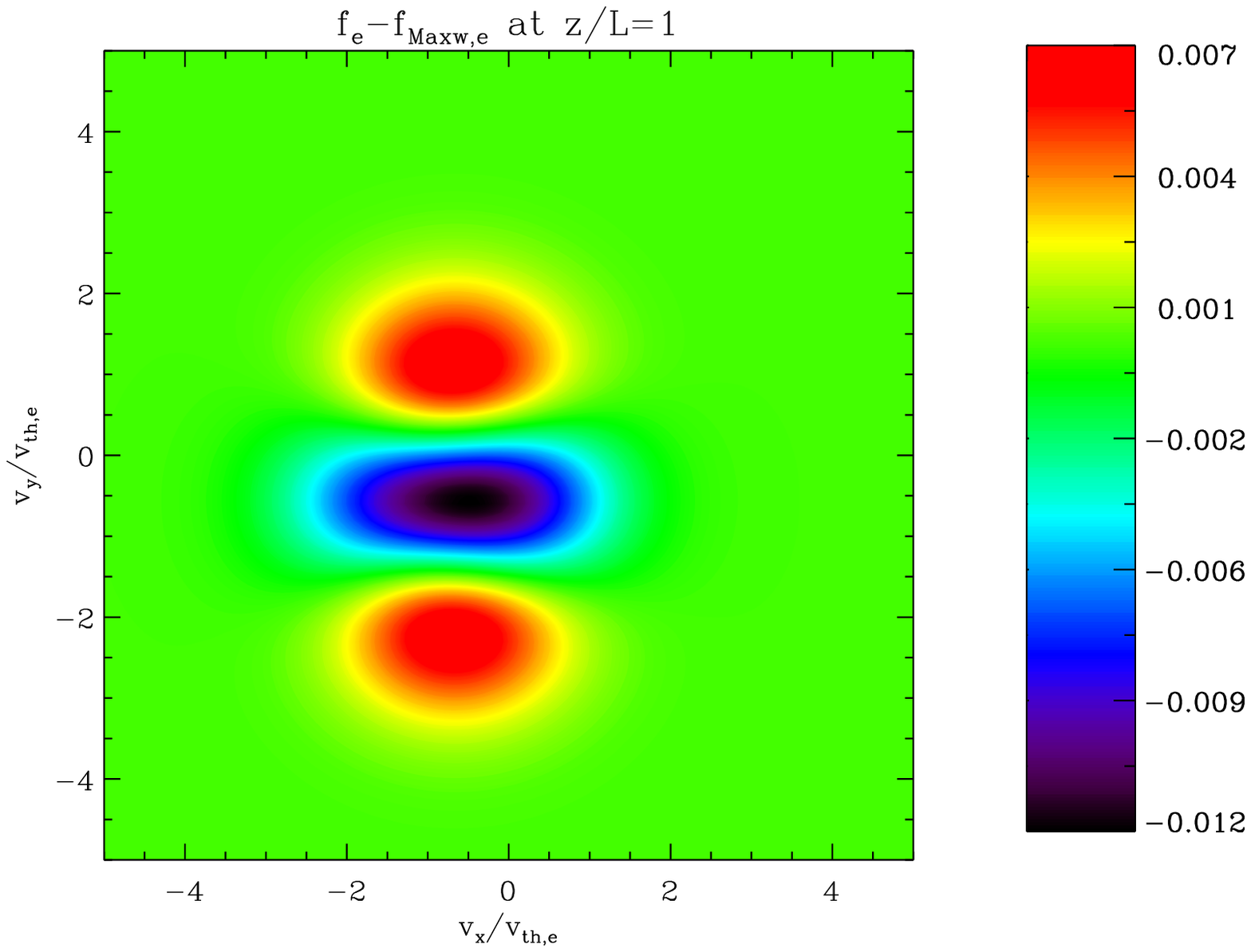}
 \caption{\label{fig:2e}}
\end{subfigure}
\caption{Contour plots of $f_e-f_{Maxw,e}$ for $z/L=-1$ (\ref{fig:2a}),  $z/L=-0.5$ (\ref{fig:2b}), $z/L=0$ (\ref{fig:2c}), $z/L=0.5$ (\ref{fig:2d}) and $z/L=1$ (\ref{fig:2e}). $\beta_{pl}=0.05$ and $\delta_e=0.03$. }
 \end{figure}

\begin{figure}
\centering
\begin{subfigure}[b]{0.48\textwidth}
\includegraphics[width=\textwidth]{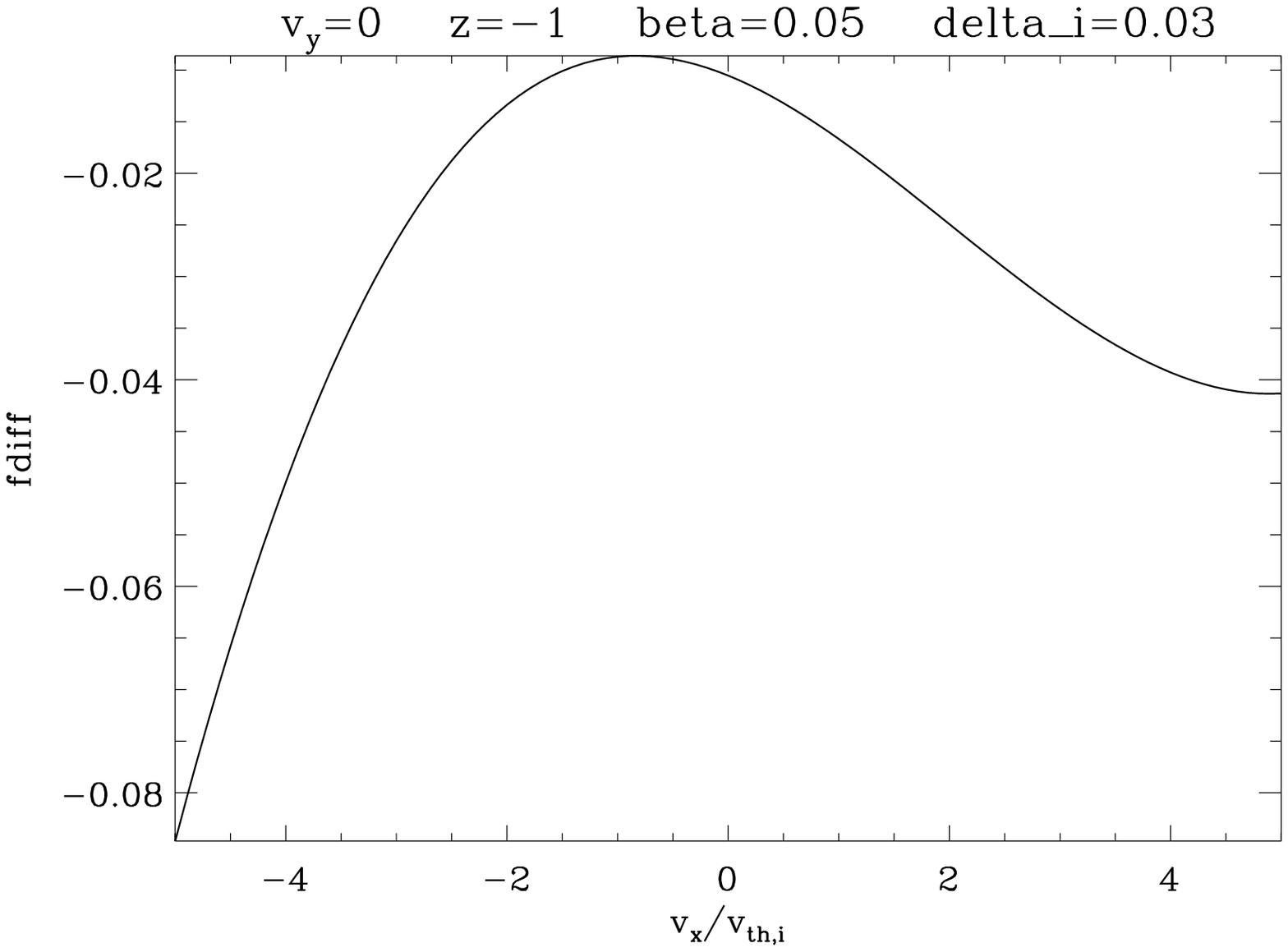}
 \caption{\label{fig:3a}}
\end{subfigure}
\begin{subfigure}[b]{0.48\textwidth}
\includegraphics[width=\textwidth]{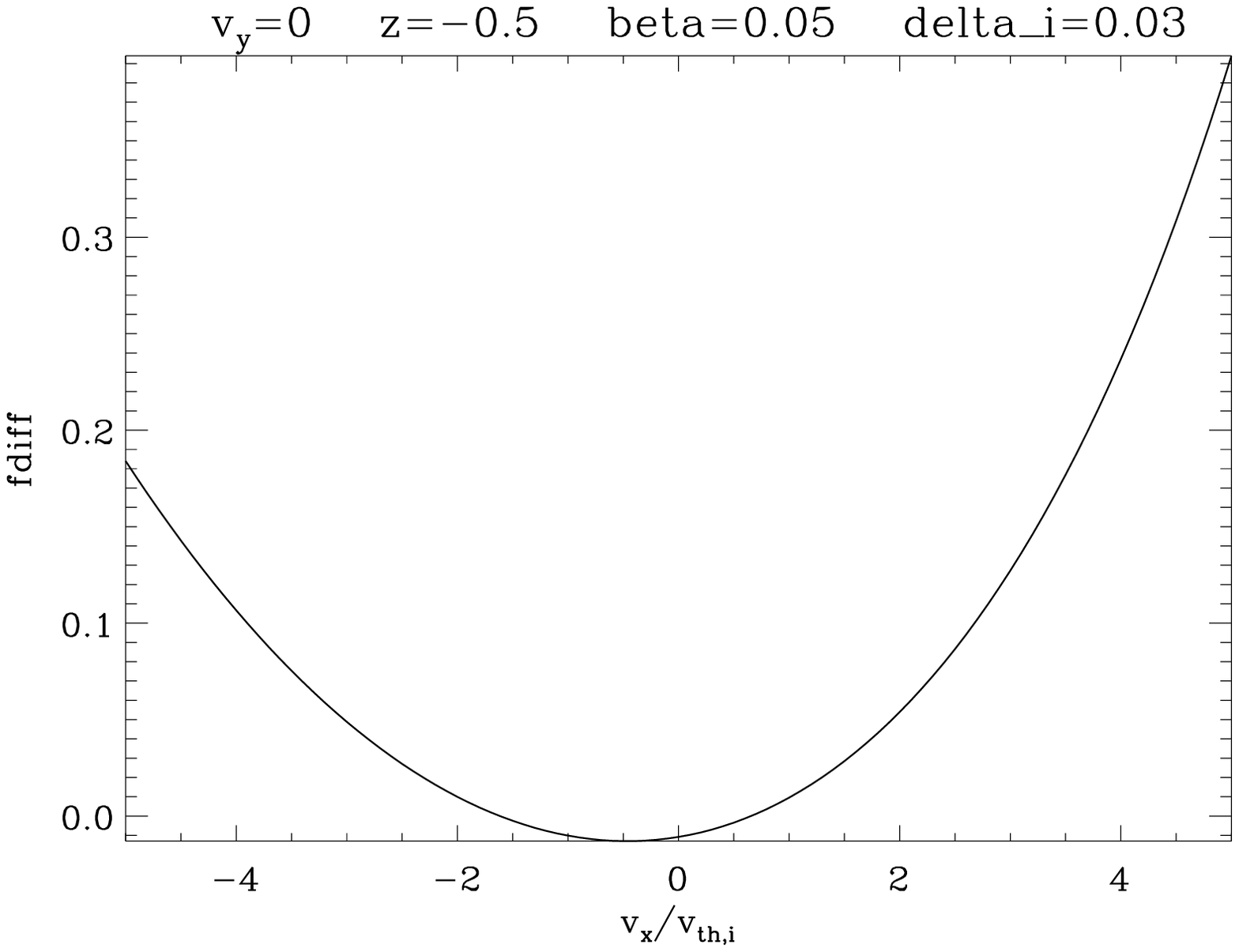}
 \caption{\label{fig:3b}}
\end{subfigure}
\begin{subfigure}[b]{0.48\textwidth}
\includegraphics[width=\textwidth]{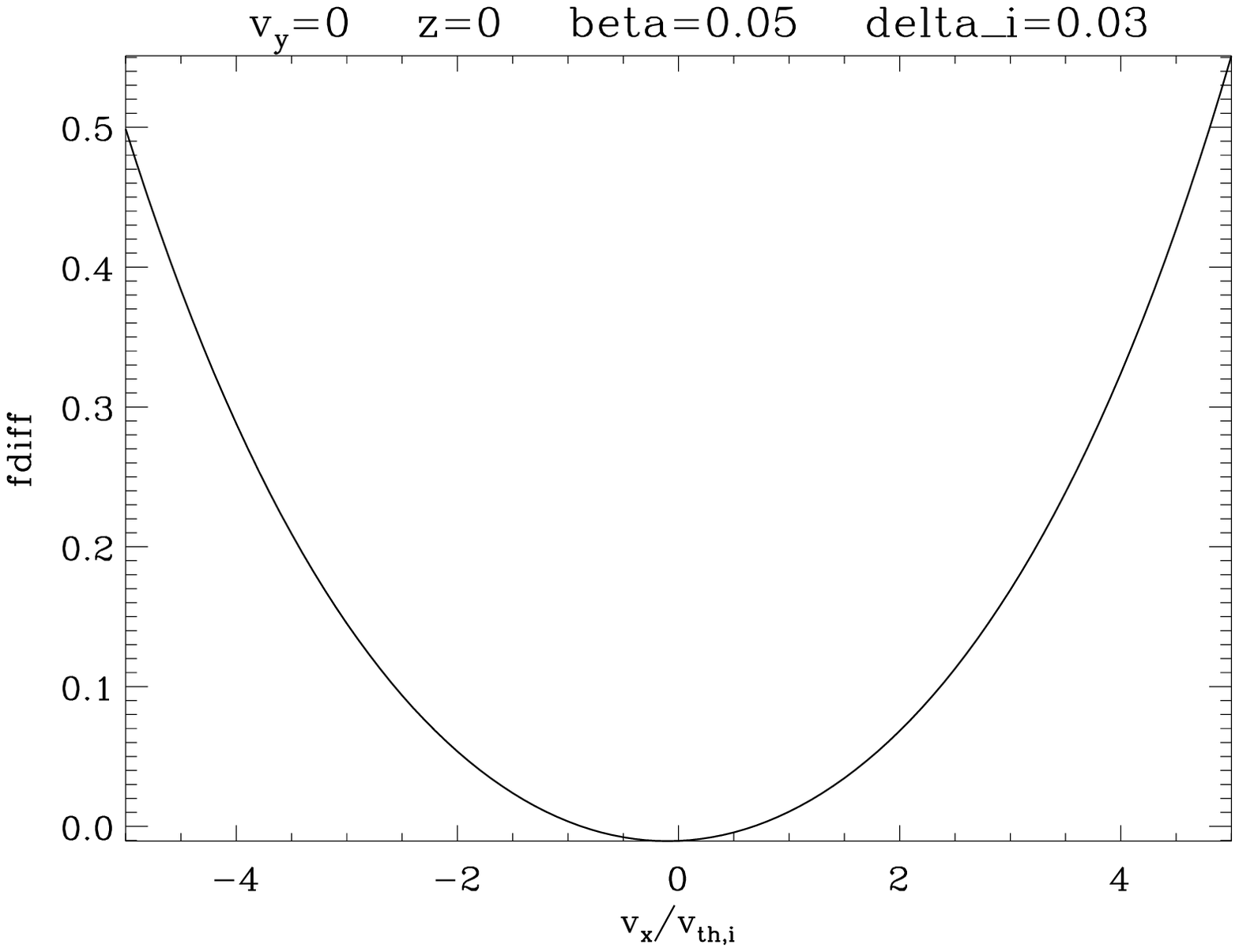}
 \caption{\label{fig:3c}}
\end{subfigure}
\begin{subfigure}[b]{0.48\textwidth}
\includegraphics[width=\textwidth]{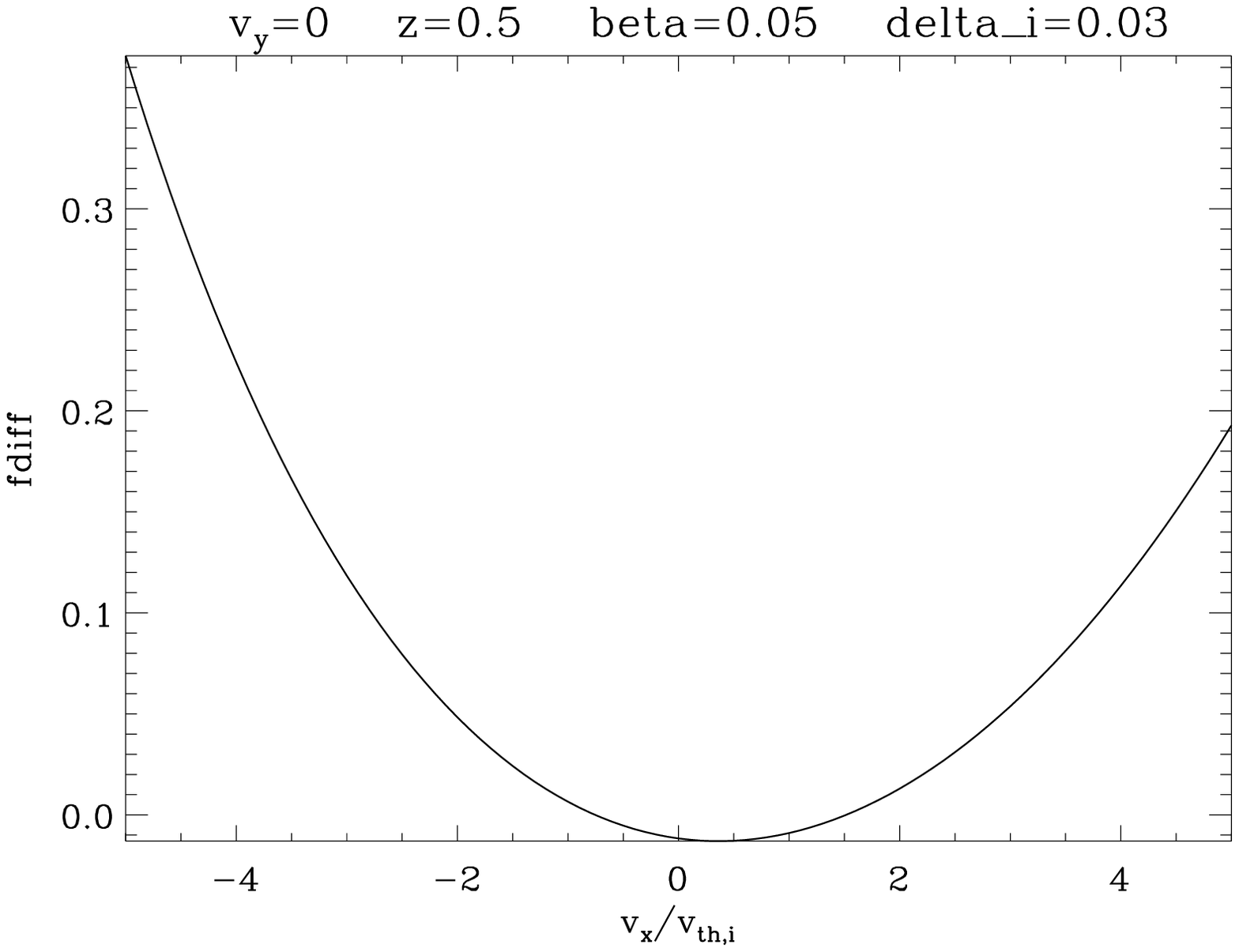}
 \caption{\label{fig:3d}}
\end{subfigure}
\begin{subfigure}[b]{0.48\textwidth}
\includegraphics[width=\textwidth]{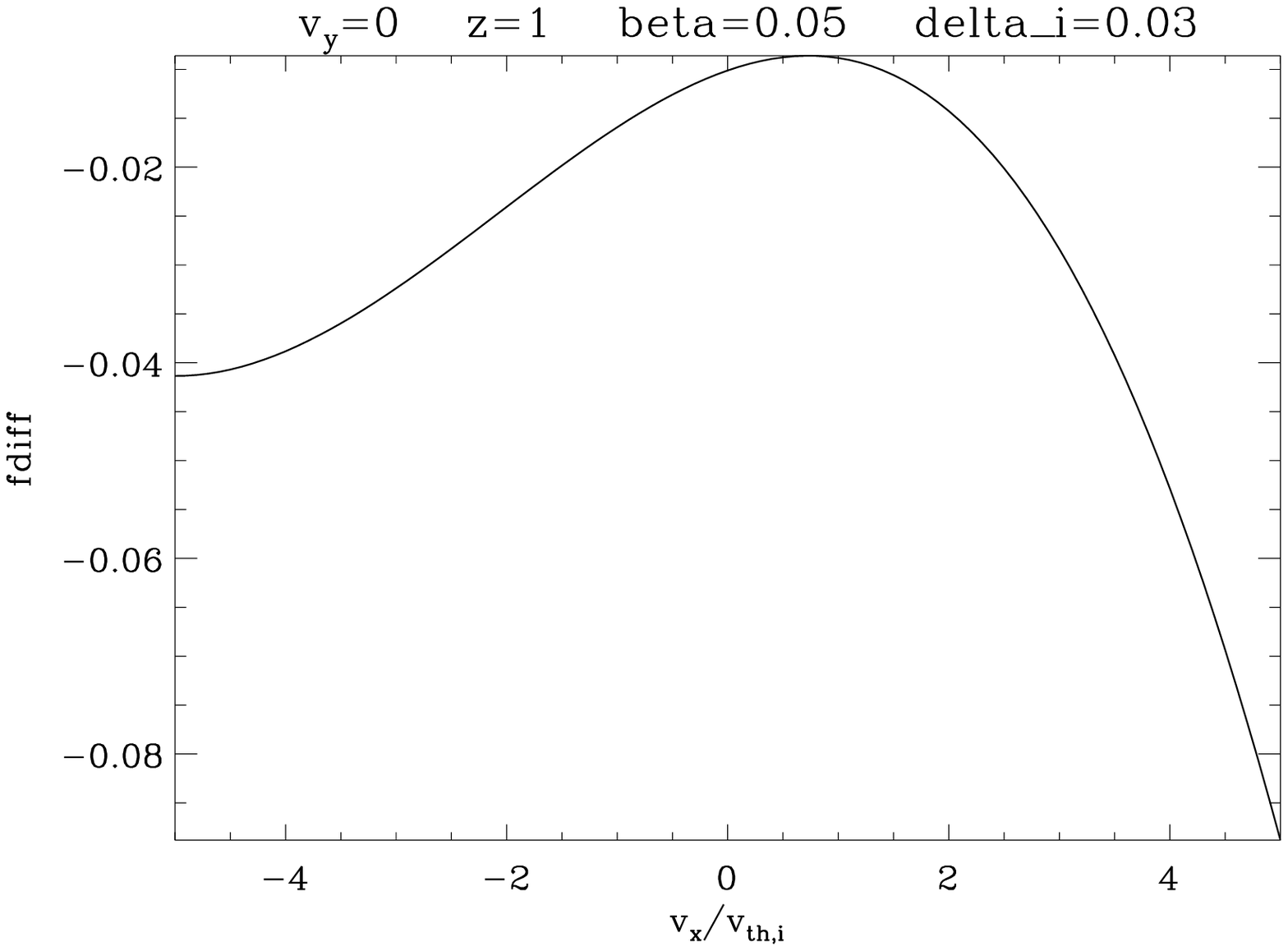}
 \caption{\label{fig:3e}}
\end{subfigure}
\caption{Line plots of $f_{diff,i}$ against $v_x/v_{th,i}$ at $v_y=0$ for $z/L=-1$ (\ref{fig:3a}),  $z/L=-0.5$ (\ref{fig:3b}), $z/L=0$ (\ref{fig:3c}), $z/L=0.5$ (\ref{fig:3d}) and $z/L=1$ (\ref{fig:3e}). $\beta_{pl}=0.05$ and $\delta_i=0.03$. }
 \end{figure}

\begin{figure}
\centering
\begin{subfigure}[b]{0.48\textwidth}
\includegraphics[width=\textwidth]{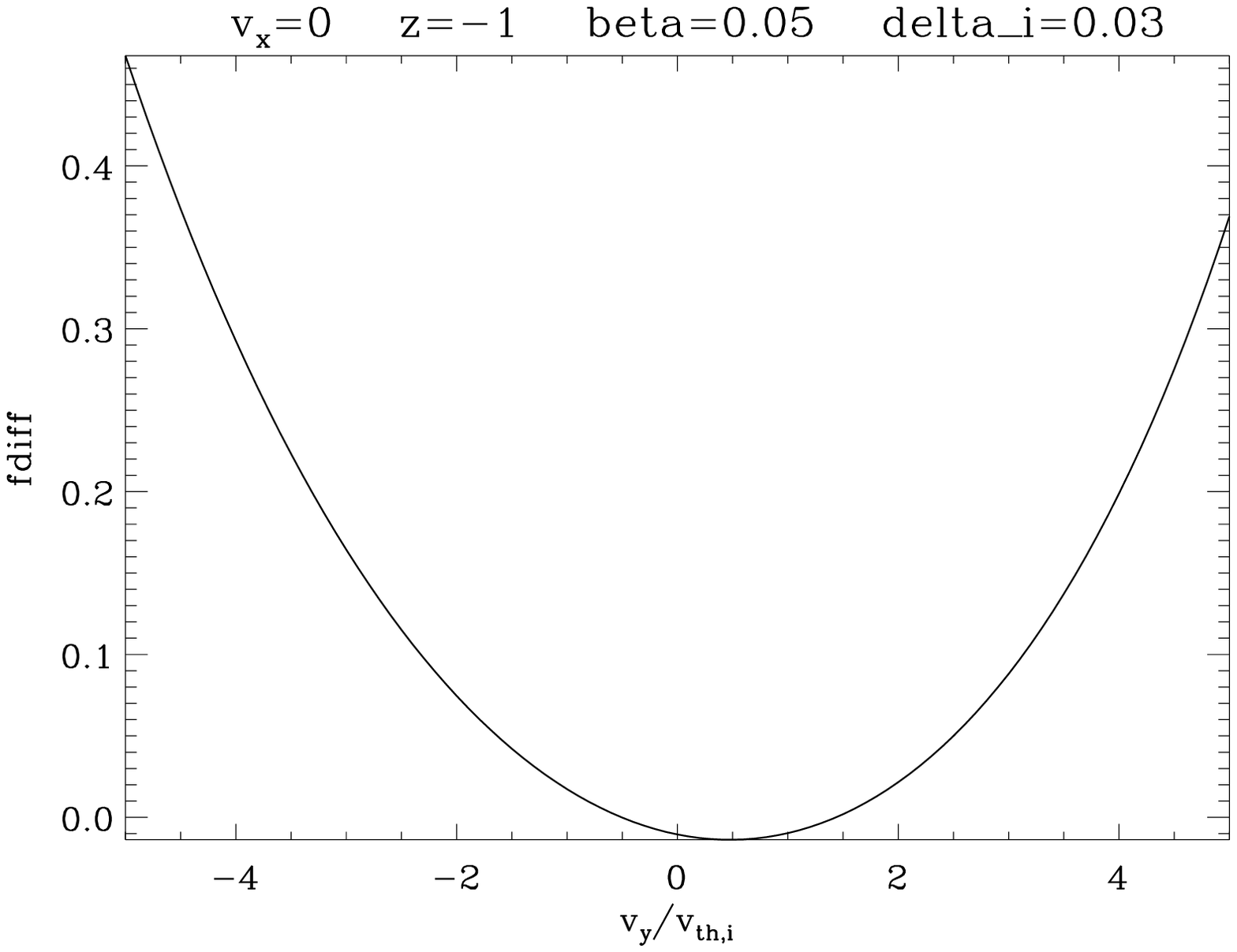}
 \caption{\label{fig:4a}}
\end{subfigure}
\begin{subfigure}[b]{0.48\textwidth}
\includegraphics[width=\textwidth]{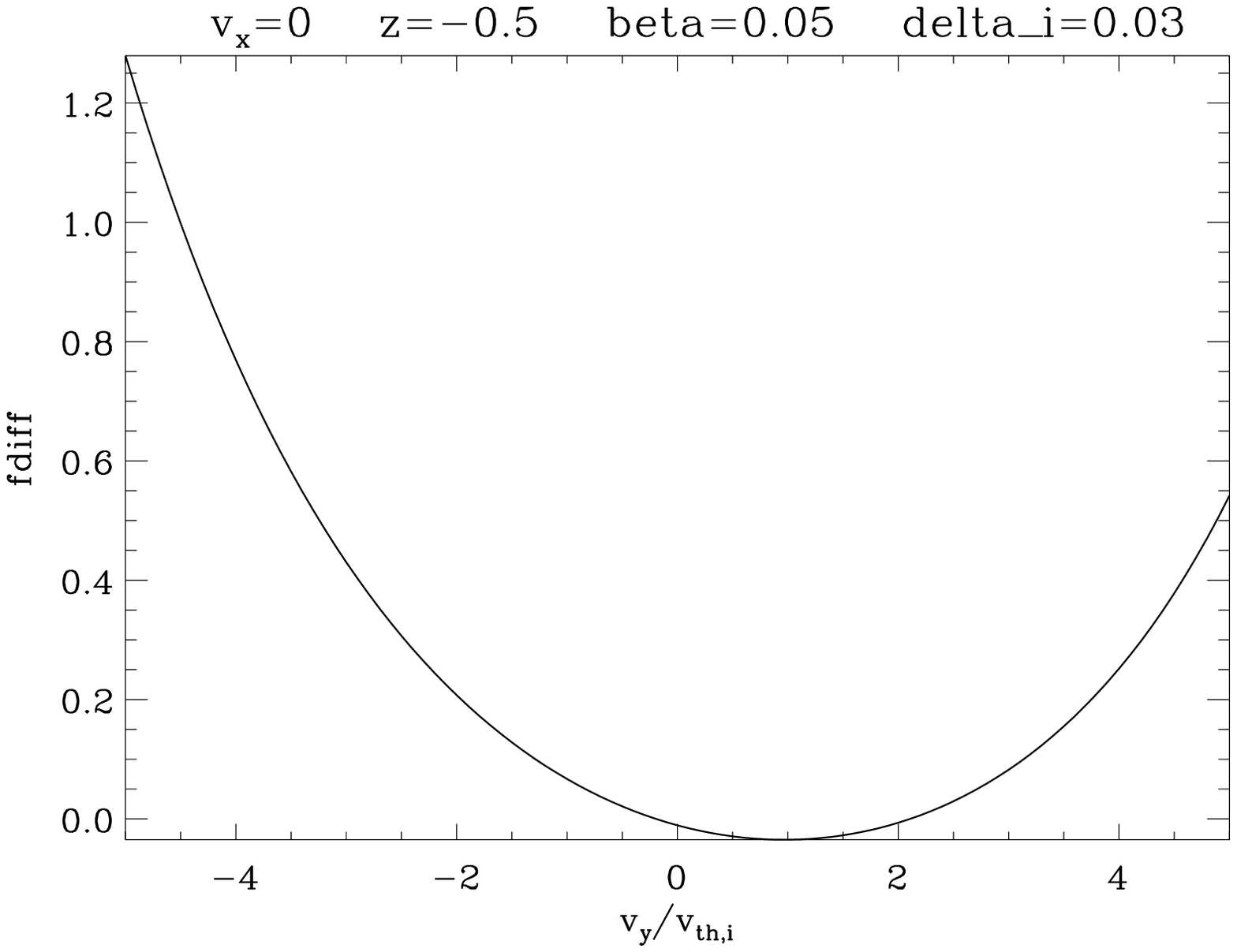}
 \caption{\label{fig:4b}}
\end{subfigure}
\begin{subfigure}[b]{0.48\textwidth}
\includegraphics[width=\textwidth]{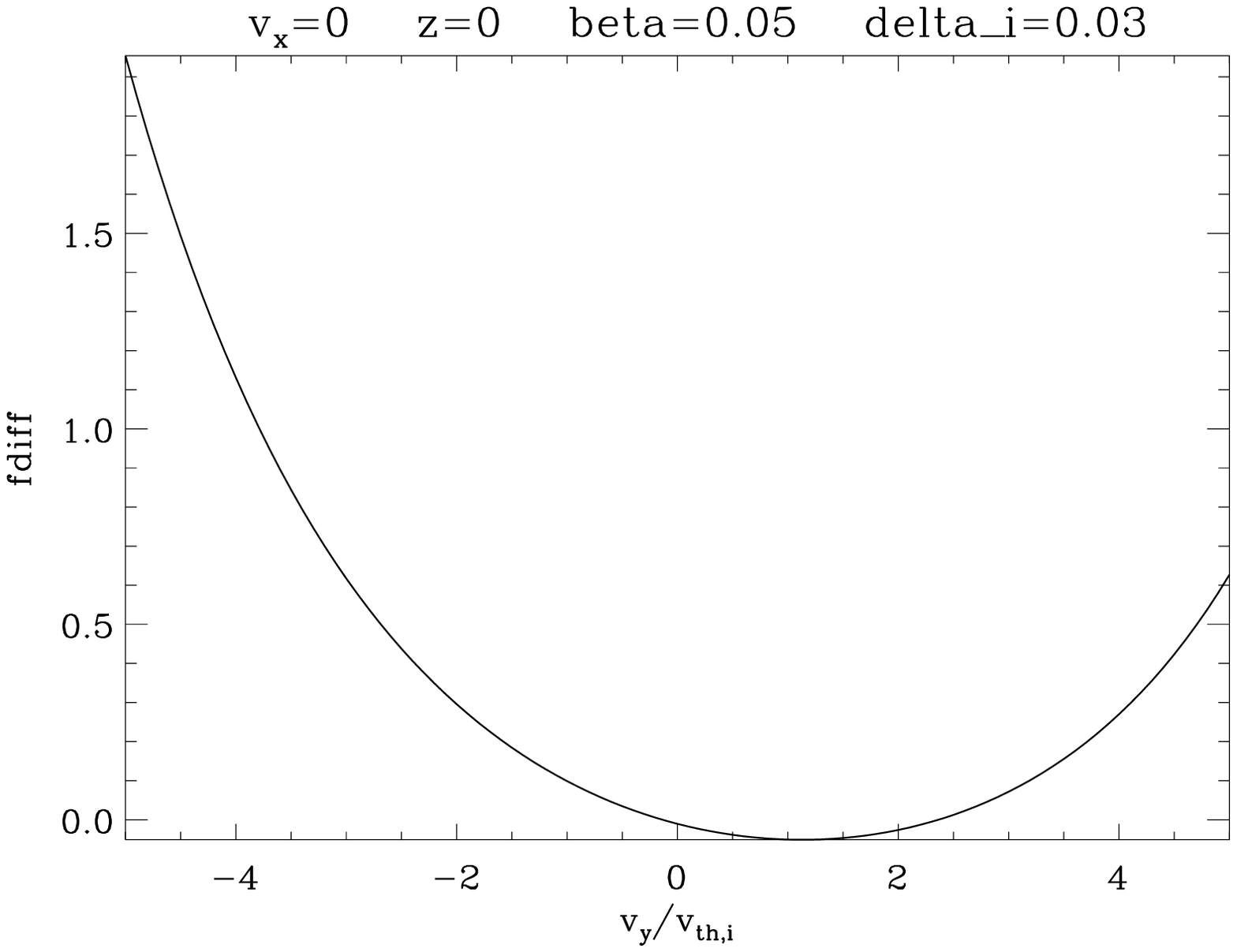}
 \caption{\label{fig:4c}}
\end{subfigure}
\begin{subfigure}[b]{0.48\textwidth}
\includegraphics[width=\textwidth]{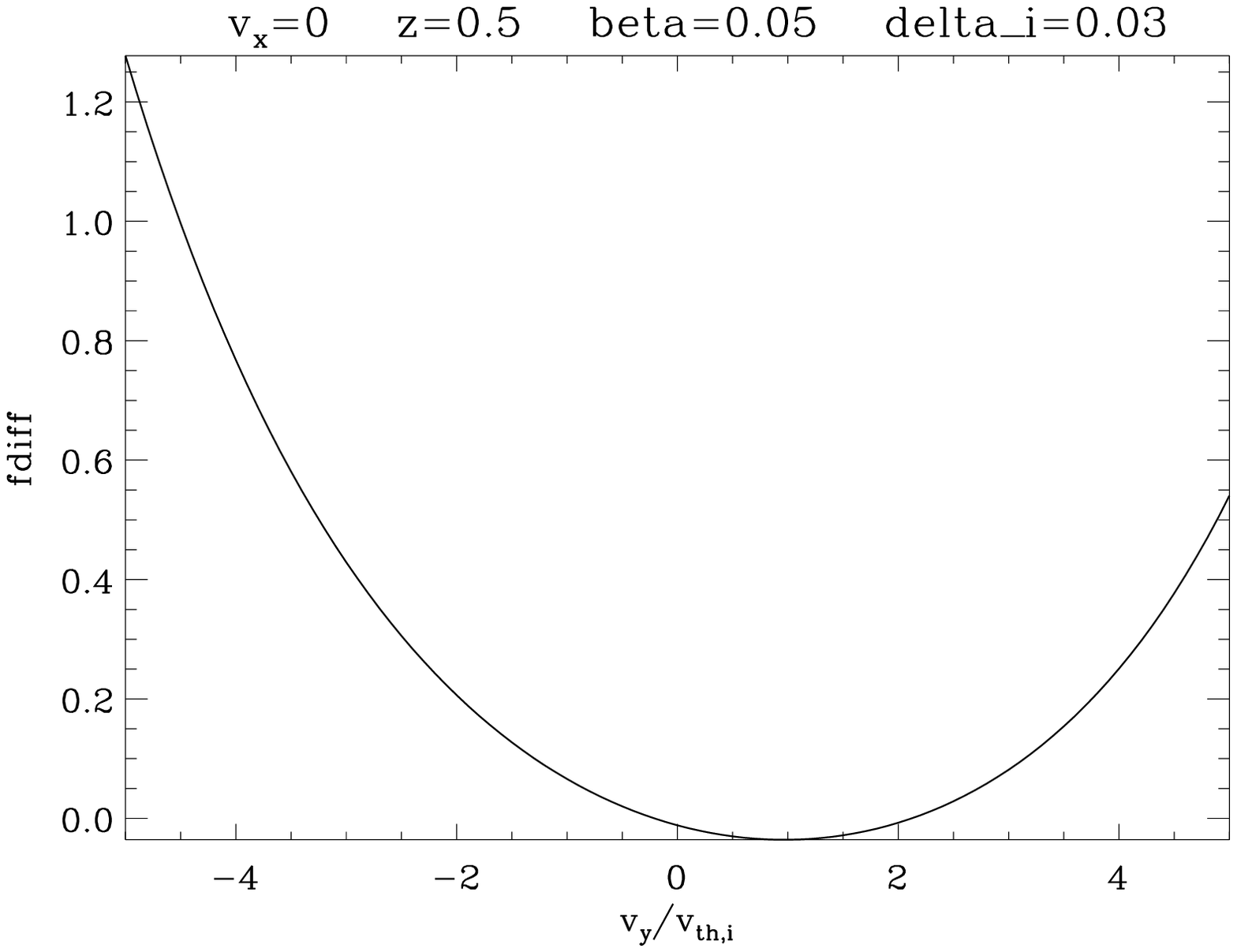}
 \caption{\label{fig:4d}}
\end{subfigure}
\begin{subfigure}[b]{0.48\textwidth}
\includegraphics[width=\textwidth]{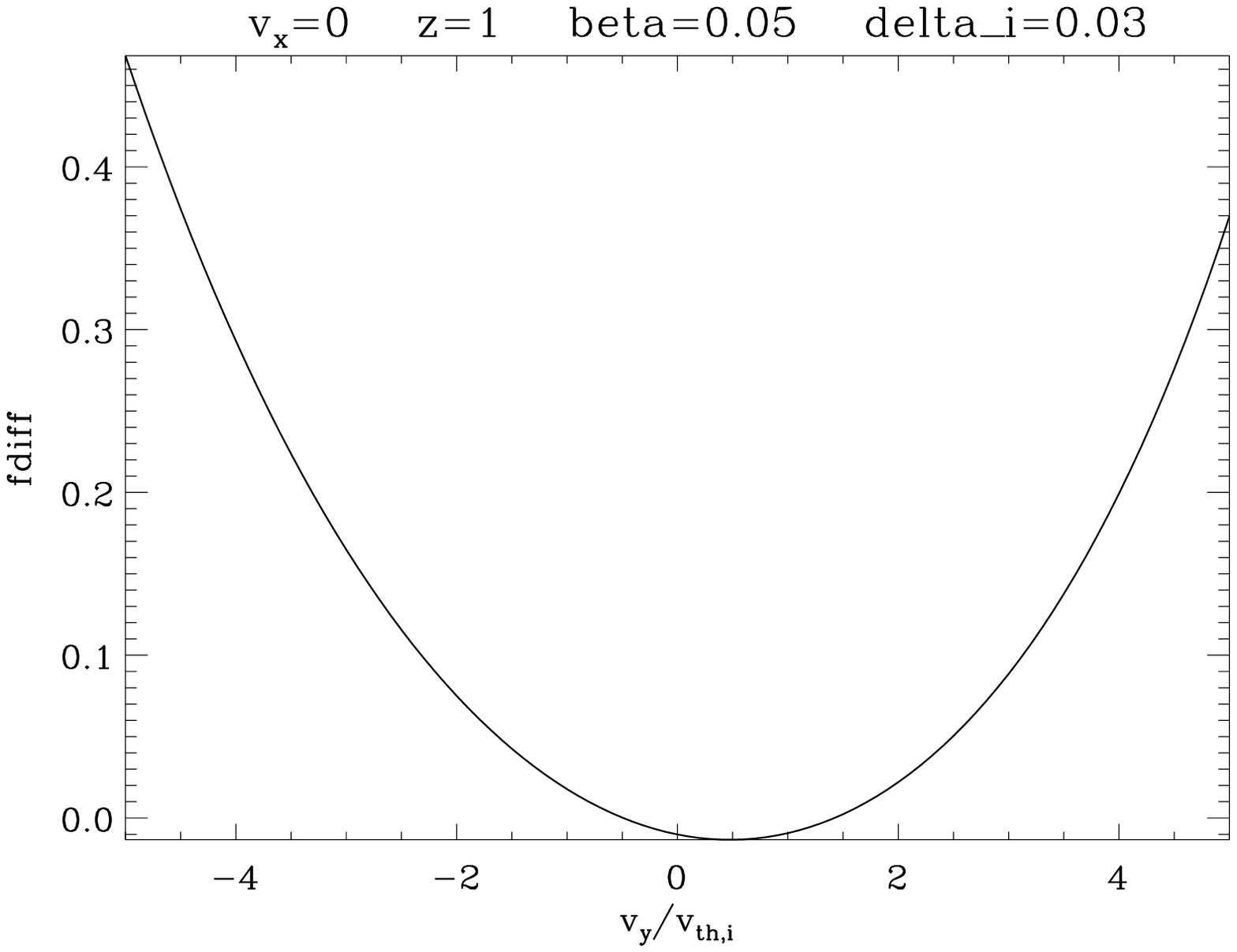}
 \caption{\label{fig:4e}}
\end{subfigure}
\caption{Line plots of $f_{diff,i}$ against $v_y/v_{th,i}$ at $v_x=0$ for $z/L=-1$ (\ref{fig:4a}),  $z/L=-0.5$ (\ref{fig:4b}), $z/L=0$ (\ref{fig:4c}), $z/L=0.5$ (\ref{fig:4d}) and $z/L=1$ (\ref{fig:4e}). $\beta_{pl}=0.05$ and $\delta_i=0.03$. }
 \end{figure}

\end{document}